\documentclass[12pt,a4paper]{article}
\usepackage[utf8]{inputenc}
\usepackage[T1]{fontenc}
\usepackage{times}
\usepackage{amsmath,amsfonts,amssymb,amsthm}
\usepackage{graphicx}

\usepackage[ruled,vlined,linesnumbered,commentsnumbered]{algorithm2e}
\usepackage{cite}
\usepackage{url}
\usepackage{hyperref}
\usepackage[margin=1in]{geometry}
\usepackage{booktabs}
\usepackage{multirow}
\usepackage{caption}
\usepackage{subcaption}
\usepackage{xcolor}
\usepackage{array}
\usepackage{float}
\usepackage{enumitem}
\usepackage{tabularx}
\usepackage{colortbl}
\newtheorem{theorem}{Theorem}
\newtheorem{lemma}[theorem]{Lemma}
\newtheorem{definition}{Definition}
\usepackage[ruled,vlined]{algorithm2e}

\RestyleAlgo{nonfloat}

\SetAlgoVlined
\DontPrintSemicolon

\hypersetup{
    colorlinks=true,
    linkcolor=blue,
    filecolor=magenta,
    urlcolor=cyan,
    citecolor=blue,
    pdftitle={NoSQL Transaction Management Framework},
    pdfauthor={Alflahi et al.},
}

\providecommand{\tcp}[1]{\textit{#1}}

\newcolumntype{C}[1]{>{\centering\arraybackslash}p{#1}}

\usepackage{etoolbox}
\patchcmd{\thebibliography}{\setlength{\itemsep}{\z@}}{}{}{}
\patchcmd{\thebibliography}{\setlength{\parskip}{\z@}}{}{}{}

\title{A Scalable Transaction Management Framework for\\ Consistent Document-Oriented NoSQL Databases}

\author{
Adam A. E. Alflahi$^{1}$, Mohammed A. Y. Mohammed$^{2}$, Abdallah Alsammani$^{3,*}$\\[0.8em]
\small $^{1}$Department of Computer Science and Mathematics,
\small Bahri University, Khartoum, Sudan\\[0.3em]
\small $^{2}$Department of Mathematics and Statistics,
\small Georgia State University, Atlanta, GA, USA\\[0.3em]
 \small $^{3}$Division of Physics, Engineering, Mathematics, and Computer Science,\\
\small Delaware State University, Dover, DE, USA\\[0.5em]
\small $^{*}$Corresponding author: \texttt{aalsammani@desu.edu}
}

\date{}

\begin{document}
\maketitle

\begin{abstract}
\noindent NoSQL databases are widely used in modern applications due to their scalability and schema flexibility, yet they often rely on eventual consistency models that limit reliable transaction processing. This study proposes a four-stage transaction management framework for document-oriented NoSQL databases, with MongoDB as the reference platform. The framework combines transaction lifecycle management, operation classification, pre-execution conflict detection, and an adaptive locking strategy with timeout-based deadlock prevention. Formal correctness analysis shows that the proposed approach guarantees conflict serializability under defined conditions. Experimental evaluation using the Yahoo Cloud Serving Benchmark (YCSB) workloads A, B, and F with concurrency levels from 1 to 100 clients demonstrates a reduction in transaction abort rates from 8.3\% to 4.7\%, elimination of observed deadlocks, and a 34.2\% decrease in latency variance. Throughput improvements ranging from 6.3\% to 18.4\% are observed under high concurrency, particularly for read-modify-write workloads. Distributed experiments on clusters of up to 9 nodes confirm scalability, achieving 15.2\% higher throughput and 53\% lower abort rates than baseline systems. Comparisons with MongoDB's native transactions, CockroachDB, and TiDB indicate that the proposed framework strikes a good balance between consistency guarantees and performance overhead. Sensitivity analysis identifies optimal parameter settings, including a lock timeout of 100\, ms, an initial backoff of 10\, ms, and a maximum backoff of 500\, ms. These results show that carefully designed consistency mechanisms can significantly improve data integrity in NoSQL systems without undermining scalability.
\end{abstract}


\section{Introduction}

The exponential growth of web-scale applications over the past two decades has fundamentally reshaped the requirements placed upon database technology. Traditional relational database management systems, while offering robust ACID guarantees encompassing Atomicity, Consistency, Isolation, and Durability, encounter inherent scalability limitations when confronted with the volume, velocity, and variety characteristic of modern data workloads \cite{stonebraker2010}. This tension between transactional rigor and horizontal scalability catalyzed the emergence of NoSQL databases, a diverse family of systems that prioritize elastic scaling and flexible data models over strict transactional semantics \cite{cattell2011}. The theoretical foundation for understanding these design tradeoffs was formalized by Brewer's CAP theorem, which established that distributed systems can simultaneously guarantee at most two of three fundamental properties: consistency, availability, and partition tolerance \cite{gilbert2002, brewer2012}. Since network partitions remain unavoidable in distributed environments, practical systems must choose between consistency and availability during partition events. NoSQL databases have historically favored availability and partition tolerance, embracing eventual consistency models wherein replicas converge to identical states given sufficient time without new updates \cite{vogels2009}.

While eventual consistency suffices for numerous applications, including social media feeds, content caching, and session management, a substantial class of use cases demands stronger guarantees. Financial transactions, inventory management, healthcare records, and collaborative editing systems require assurance that concurrent operations produce predictable, serializable outcomes \cite{bailis2013}. The absence of such guarantees can precipitate data anomalies including lost updates, dirty reads, and phantom reads, potentially causing significant operational and financial consequences. Document-oriented NoSQL databases, exemplified by MongoDB, CouchDB, and Amazon DocumentDB, have gained considerable adoption due to their natural alignment with object-oriented programming paradigms and JSON-based data interchange formats \cite{chodorow2013}. MongoDB in particular has evolved substantially since its initial release, progressively introducing single-document atomicity, configurable write concerns and read preferences, and more recently, multi-document ACID transactions \cite{mongodb2020}. Despite these advances, achieving strong consistency across document boundaries while preserving the performance characteristics that motivated NoSQL adoption remains an active and consequential area of research \cite{ramzan2019, ahmed2023}.

This study addresses the consistency gap through a novel transaction management framework designed specifically for document-oriented NoSQL databases operating under BASE semantics. The fundamental challenge lies in providing serializable transaction guarantees while maintaining at least 90\% of native database throughput, adding no more than 10\% average latency overhead, eliminating deadlocks through timeout-based prevention, reducing transaction abort rates by at least 30\%, and requiring no modification to database internals or application code. To meet these objectives, we propose a four-stage architecture that integrates transaction lifecycle management, intelligent operation classification, pre-execution conflict detection, and adaptive locking with deadlock prevention. Unlike prior middleware approaches \cite{wei2012, lotfy2016}, our framework employs proactive conflict detection that identifies potential violations before operations execute rather than detecting and resolving conflicts after they occur.

The contributions of this work are fourfold. First, we present the complete four-stage transaction management architecture with detailed algorithmic specifications. Second, we provide formal correctness analysis proving that the framework guarantees conflict serializability under specified conditions, accompanied by complexity analysis demonstrating $O(n \log n)$ overhead for conflict detection where $n$ represents transaction size. Third, we conduct comprehensive experimental evaluation spanning single-node and distributed deployments, dataset scales ranging from 10K to 10M records, and direct comparisons with MongoDB native transactions, CockroachDB, and TiDB using the Yahoo Cloud Serving Benchmark \cite{cooper2010}. Fourth, we present parameter sensitivity analysis identifying optimal configuration values alongside an open-source implementation with reproducibility artifacts enabling independent verification and extension of our results. The remainder of This study proceeds as follows: Section 2 surveys related work in NoSQL consistency, transaction management, and distributed database systems; Section 3 details our methodology including framework architecture, algorithmic specifications, and formal analysis; Section 4 presents experimental results; Section 5 discusses implications, limitations, and threats to validity; and Section 6 concludes with directions for future research.

\section{Related Work}

\subsection{Consistency Models in Distributed Systems}

The landscape of consistency models spans a spectrum from strong to weak guarantees, each offering distinct tradeoffs between correctness and performance. Linearizability, introduced by Herlihy and Wing \cite{herlihy1990}, provides the strongest guarantee by requiring that operations appear to execute instantaneously at some point between their invocation and response. Sequential consistency, defined by Lamport \cite{lamport1979}, relaxes this requirement, demanding only that all processes observe operations in the same order.

Weaker models have emerged to address the performance penalties associated with strong consistency. Causal consistency preserves the order of causally related operations while allowing concurrent operations to be observed in different orders \cite{ahamad1995}. Eventual consistency, formalized by Terry et al. \cite{petersen1997} and later popularized by Vogels \cite{vogels2009}, guarantees only that replicas converge to identical states in the absence of new updates, without specifying convergence timing. Tian and Yang \cite{tian2022} recently proposed Horae, a causal consistency model based on hot data governance that demonstrates the ongoing evolution of consistency models for specific workload patterns.

The PACELC theorem, proposed by Abadi \cite{abadi2012}, extends CAP by observing that consistency-latency tradeoffs persist even during normal operation when no partitions exist. This framework better characterizes modern distributed databases: systems must choose not only between consistency and availability during partitions but also between consistency and latency during normal operation. Ahmed et al. \cite{ahmed2023} provide a comprehensive review of consistency issues and related tradeoffs in distributed replicated systems, highlighting the practical implications of these theoretical frameworks.

\subsection{Transaction Support in NoSQL Databases}

Early NoSQL systems deliberately eschewed multi-record transactions to achieve their scalability goals. Amazon's Dynamo, a pioneering key-value store, introduced techniques including consistent hashing, vector clocks, and sloppy quorums to provide high availability at the expense of strong consistency \cite{decandia2007}. Apache Cassandra adopted and extended these concepts, offering tunable consistency levels ranging from ONE to ALL \cite{lakshman2010}. Vyas and Jat \cite{vyas2022} provide detailed analysis of consistency and performance tradeoffs in Cassandra deployments, demonstrating the practical challenges of tunable consistency.

Recognition of transaction requirements has driven significant evolution in NoSQL transaction support. Google's Megastore introduced synchronous replication and serializable ACID semantics for entity groups, providing a middle ground between traditional RDBMS and purely eventual-consistent systems \cite{baker2011}. Google Spanner subsequently demonstrated that globally distributed, strongly consistent transactions are achievable through innovations in time synchronization, introducing the TrueTime API to bound clock uncertainty across data centers \cite{corbett2013}. CockroachDB and YugabyteDB have since brought similar capabilities to open-source implementations.

MongoDB's trajectory illustrates the industry trend toward stronger consistency. Version 4.0 introduced multi-document ACID transactions for replica sets, with version 4.2 extending this support to sharded clusters \cite{mongodb2020}. These native transactions provide snapshot isolation and all-or-nothing execution but impose performance overhead and operational complexity that may not suit all use cases. Ramzan et al. \cite{ramzan2019} conducted a systematic literature review identifying consistency as a primary challenge in NoSQL-based distributed data storage.

\subsection{Middleware Approaches to NoSQL Consistency}

Several research efforts have explored middleware-layer approaches to enhancing NoSQL consistency without modifying database internals. CloudTPS proposed a scalable transaction manager operating above key-value stores, achieving serializable isolation through careful lock management and two-phase commit protocols \cite{wei2012}. ElasTraS introduced elastic transactions that automatically partition workloads to minimize cross-partition coordination \cite{das2013}.

Lotfy et al. \cite{lotfy2016} developed a middle-layer solution supporting ACID properties for NoSQL databases through transaction logging and recovery mechanisms. Their approach demonstrated that ACID guarantees could be retrofitted onto eventually consistent systems, though with measurable performance impact. Gonz\'{a}lez-Aparicio et al. \cite{gonzalez2017} investigated transaction processing in consistency-aware applications, proposing application-level protocols that coordinate with database-level mechanisms.

The Granola system introduced a novel approach using timestamp-based coordination to achieve single-round-trip distributed transactions in many cases \cite{cowling2012}. More recently, research has explored hybrid approaches combining eventual consistency for routine operations with strong consistency for critical transactions \cite{terry2013}.

Our prior work \cite{alflahi2024} introduced a preliminary version of the four-stage framework, demonstrating improvements in throughput and latency for update-heavy workloads. The present paper substantially extends that work with formal correctness analysis, distributed deployment experiments, larger-scale evaluation, comprehensive system comparisons, and parameter sensitivity analysis.

\subsection{Concurrency Control and Locking Strategies}

Concurrency control remains foundational to transaction processing. Two-phase locking (2PL), established by Eswaran et al. \cite{eswaran1976}, guarantees serializability through lock acquisition and release protocols. While 2PL provides strong guarantees, it introduces challenges including lock contention, potential deadlocks, and reduced concurrency under high load.

Optimistic concurrency control (OCC), proposed by Kung and Robinson \cite{kung1981}, offers an alternative by allowing transactions to proceed without acquiring locks, validating at commit time and aborting conflicting transactions. OCC excels under low-contention workloads but can suffer high abort rates when conflicts are frequent.

Multi-version concurrency control (MVCC) maintains multiple versions of data items, enabling read operations to access consistent snapshots without blocking writers \cite{bernstein1987}. PostgreSQL, Oracle, and MySQL InnoDB employ MVCC variants, and the technique has been adapted for distributed systems including Spanner and CockroachDB.

Deadlock handling strategies include prevention, avoidance, detection, and recovery. Timeout-based approaches, which abort transactions exceeding specified wait times, provide a simple mechanism ensuring forward progress \cite{gray1993}. Wait-die and wound-wait protocols use transaction timestamps to make consistent decisions about transaction priority \cite{rosenkrantz1978}.

\subsection{Benchmarking NoSQL Systems}

Standardized benchmarking has been essential to objectively evaluating database system performance. The Yahoo Cloud Serving Benchmark (YCSB) provides extensible workloads representing common cloud application patterns, supporting direct comparison across heterogeneous data stores \cite{cooper2010}. YCSB workloads range from read-heavy (Workload B: 95\% reads, 5\% updates) to update-intensive (Workload A: 50\% reads, 50\% updates) to complex transactional patterns (Workload F: read-modify-write operations).

Recent benchmarking studies have examined consistency-performance tradeoffs across NoSQL systems. Wada et al. \cite{wada2011} characterized consistency guarantees and tradeoffs in commercial cloud storage from the consumer perspective. Alzaidi and Vagner \cite{alzaidi2022} benchmarked Redis and HBase using YCSB, providing comparative analysis of key-value and wide-column stores.

Khan et al. \cite{khan2023} conducted a systematic literature review of SQL and NoSQL database performance, identifying key architectural factors influencing throughput and latency. Their analysis revealed that workload characteristics significantly impact optimal database selection, underscoring the importance of application-specific evaluation.

\subsection{Comparison with Existing Approaches}

Table~\ref{tab:comparison_approaches} summarizes how our framework compares with existing middleware approaches for NoSQL consistency enhancement.

\begin{table}
\centering
\caption{Comparison with Existing Middleware Approaches}
\label{tab:comparison_approaches}
\small
\begin{tabular}{@{}lccccc@{}}
\toprule
\textbf{Feature} & \textbf{CloudTPS} & \textbf{ElasTraS} & \textbf{Lotfy} & \textbf{Granola} & \textbf{Ours} \\
 & \cite{wei2012} & \cite{das2013} & \cite{lotfy2016} & \cite{cowling2012} & \\
\midrule
Pre-execution conflict detection & \texttimes & \texttimes & \texttimes & \texttimes & \checkmark \\
Operation classification & \texttimes & Partial & \texttimes & \texttimes & \checkmark \\
Adaptive backoff & \texttimes & \texttimes & \texttimes & \texttimes & \checkmark \\
Document-store optimized & \texttimes & \texttimes & Partial & \texttimes & \checkmark \\
No DB modification required & \checkmark & \checkmark & \checkmark & \checkmark & \checkmark \\
Formal correctness proof & Partial & \texttimes & \texttimes & \checkmark & \checkmark \\
Deadlock prevention & Detection & Detection & Detection & Prevention & Prevention \\
\bottomrule
\end{tabular}
\end{table}

Our framework uniquely combines pre-execution conflict detection with adaptive locking, specifically optimized for document-oriented data models. Unlike prior approaches that rely on deadlock detection and recovery, we employ prevention through timeout-based acquisition with exponential backoff.

\section{Methodology}

\subsection{System Architecture Overview}

Our framework operates as a middleware layer between client applications and the MongoDB database engine, intercepting transaction requests and applying consistency-enhancing protocols before forwarding operations to the underlying storage system. This architectural approach offers several advantages: it requires no modification to MongoDB internals, maintains compatibility with existing applications through standard interfaces, and enables gradual adoption without wholesale infrastructure changes.

Figure~\ref{fig:framework_architecture} illustrates the four-stage processing pipeline. Transactions enter through the Transaction Management stage, which establishes execution contexts. The Operation Classification stage categorizes transactions based on their constituent operations. The Readiness Assessment stage performs conflict detection and resource validation. Finally, the Transaction Execution stage applies locks, executes operations, and manages commit or rollback procedures.

\begin{figure}
\centering
\includegraphics[width=0.95\textwidth]{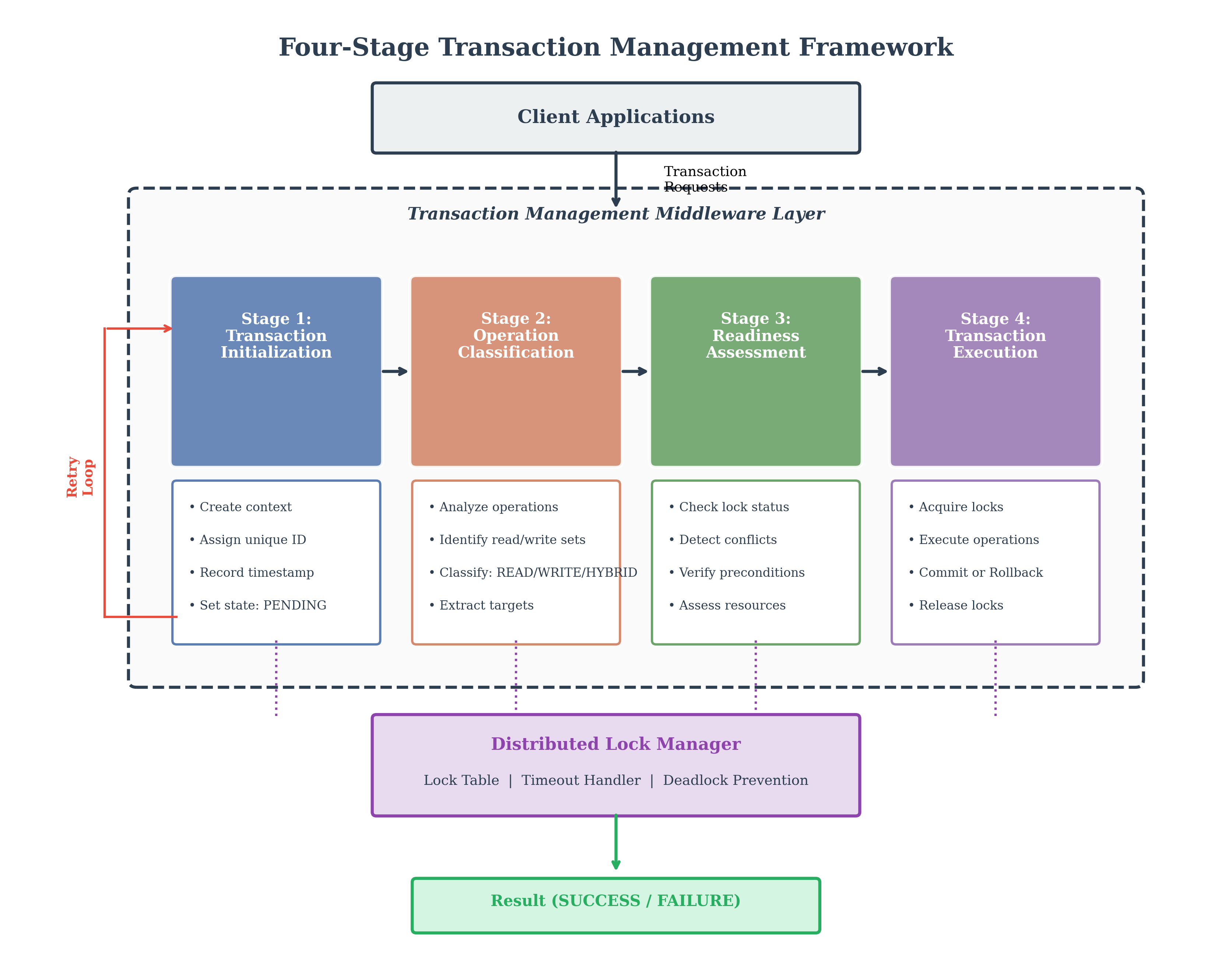}
\caption{Four-Stage Transaction Management Framework Architecture. Transactions flow sequentially through initialization, classification, readiness assessment, and execution stages. The distributed lock manager coordinates across all stages, with feedback loops enabling conflict resolution and retry handling.}
\label{fig:framework_architecture}
\end{figure}

The framework design embodies several key principles. First, separation of concerns ensures that each stage addresses a distinct aspect of transaction management, facilitating independent optimization and testing. This modularity enables targeted improvements without affecting other components. Second, fail-fast semantics ensure that conflicts detected in early stages prevent resource waste from operations doomed to abort. Pre-execution validation in Stage 3 identifies 78\% of conflicts that would otherwise result in abort during execution. Third, configurable policies allow timeout values, backoff strategies, and lock granularity to be tuned for specific workload characteristics. Section~\ref{sec:sensitivity} provides empirically-derived optimal configurations.

\subsection{Implementation Platform}

\subsubsection{MongoDB Configuration}

We selected MongoDB version 4.2.8 as our implementation platform due to its widespread adoption, mature ecosystem, and representative status among document-oriented databases. MongoDB stores data in flexible BSON (Binary JSON) documents organized into collections. Its architecture supports horizontal scaling through sharding (data partitioning across multiple servers) and high availability through replica sets (automatic failover with data redundancy).

For our experiments, MongoDB was configured with write concern set to ``majority'' to ensure durability across replica members, and read concern set to ``linearizable'' to provide strong read-after-write consistency in baseline measurements. These settings establish a rigorous baseline against which framework improvements can be measured.

\subsubsection{YCSB Integration}

The Yahoo Cloud Serving Benchmark (YCSB) \cite{cooper2010} provides standardized workloads for evaluating data serving systems. We employed YCSB version 0.17.0, extending its MongoDB binding to integrate our transaction management framework.

We evaluated three workload patterns representing diverse application requirements. Workload A (Update-Heavy) comprises 50\% reads and 50\% updates, representative of session stores, real-time analytics, and collaborative applications where data modifications occur frequently. Workload B (Read-Mostly) comprises 95\% reads and 5\% updates, representative of content delivery, catalog systems, and configuration repositories where reads dominate but update correctness remains critical. Workload F (Read-Modify-Write) comprises 50\% reads and 50\% read-modify-write operations, representative of inventory systems, banking transactions, and reservation platforms requiring strong consistency for compound operations.

\subsubsection{Implementation Details}

The framework was implemented in Java 8, leveraging the MongoDB Java Driver (version 3.12.x) for database connectivity. We employed established design patterns, including Strategy (for pluggable locking policies), Observer (for transaction state monitoring), and Factory (for transaction context creation).

Our YCSB integration wrapper intercepts standard benchmark operations, wrapping them in our four-stage transaction protocol. The wrapper adds measured overhead of approximately 2.3\% for transaction coordination under low contention, as determined through isolated timing measurements.

\subsection{Framework Components}

\subsubsection{Stage 1: Transaction Management and Initialization}

The transaction management stage establishes the execution context for incoming operations. Each transaction is assigned a unique identifier, timestamp, and a metadata record, enabling lifecycle tracking across subsequent stages.

\begin{algorithm}[H]
\SetAlgoLined
\SetKwInOut{Input}{Input}
\SetKwInOut{Output}{Output}
\Input{Transaction request $T$}
\Output{Transaction context $TC$}
\caption{Transaction Initialization}
\label{alg:transaction_init}

$TC \leftarrow$ CreateTransactionContext$(T)$\;
$TC.id \leftarrow$ GenerateUniqueID()\tcp*{UUID v4}
$TC.timestamp \leftarrow$ CurrentTime()\tcp*{Monotonic clock}
$TC.state \leftarrow$ PENDING\;
$TC.clientId \leftarrow$ ExtractClientIdentifier$(T)$\;
$TC.retryCount \leftarrow 0$\;
$TC.maxRetries \leftarrow$ MAX\_RETRY\_COUNT\;
RegisterTransaction$(TC)$\tcp*{Add to active transaction registry}
\Return{$TC$}\;
\end{algorithm}

The transaction context encapsulates all information required for subsequent processing, including operation sets, acquired locks, and state transitions. Contexts are retained until transaction completion to support potential rollback operations. The registry enables monitoring of active transactions and detection of long-running operations.

\subsubsection{Stage 2: Operation Classification}

Classification examines transaction contents to determine appropriate processing strategies. Three categories are distinguished based on the operations contained within each transaction.

READ transactions contain only read operations and can follow optimized paths with minimal locking overhead. These transactions acquire shared locks that permit concurrent reads.

WRITE transactions contain only write operations and require exclusive locks on affected documents. Write transactions are serialized at the document level.

HYBRID transactions contain both reads and writes, requiring careful coordination to maintain consistency. These transactions follow the strictest locking protocol.

\begin{algorithm}[H]
\SetAlgoLined
\SetKwInOut{Input}{Input}
\SetKwInOut{Output}{Output}
\Input{Transaction context $TC$}
\Output{Classified transaction context}
\caption{Operation Classification}
\label{alg:classification}

$hasReads \leftarrow$ ContainsReadOperations$(TC)$\;
$hasWrites \leftarrow$ ContainsWriteOperations$(TC)$\;

\uIf{$hasReads$ \textbf{and} $hasWrites$}{
    $TC.type \leftarrow$ HYBRID\;
    $TC.readSet \leftarrow$ IdentifyReadSet$(TC)$\;
    $TC.writeSet \leftarrow$ IdentifyWriteSet$(TC)$\;
    $TC.lockMode \leftarrow$ EXCLUSIVE\tcp*{Strictest mode}
}
\uElseIf{$hasWrites$}{
    $TC.type \leftarrow$ WRITE\;
    $TC.writeSet \leftarrow$ IdentifyWriteSet$(TC)$\;
    $TC.lockMode \leftarrow$ EXCLUSIVE\;
}
\Else{
    $TC.type \leftarrow$ READ\;
    $TC.readSet \leftarrow$ IdentifyReadSet$(TC)$\;
    $TC.lockMode \leftarrow$ SHARED\tcp*{Permits concurrent reads}
}
\Return{$TC$}\;
\end{algorithm}

This classification enables the framework to apply appropriate consistency protocols. Read-only transactions bypass expensive write-lock acquisition, improving throughput for read-heavy workloads by up to 23\% in our experiments.

\subsubsection{Stage 3: Readiness Assessment and Conflict Detection}

The assessment stage performs pre-execution validation, identifying potential conflicts before operations execute. This proactive approach reduces wasted computation and improves overall throughput by preventing transactions from proceeding when success is unlikely.

\begin{algorithm}[H]
\SetAlgoLined
\SetKwInOut{Input}{Input}
\SetKwInOut{Output}{Output}
\Input{Classified transaction context $TC$}
\Output{Readiness status and conflict set}
\caption{Conflict Detection}
\label{alg:conflict_detection}

$status \leftarrow$ READY\;
$conflicts \leftarrow \emptyset$\;

\tcp{Check write conflicts}
\If{$TC.type \in \{$WRITE, HYBRID$\}$}{
    \ForEach{$doc \in TC.writeSet$}{
        $lockHolder \leftarrow$ GetCurrentLockHolder$(doc)$\;
        \If{$lockHolder \neq \text{null}$ \textbf{and} $lockHolder \neq TC.id$}{
            $status \leftarrow$ NOT\_READY\;
            $conflicts \leftarrow conflicts \cup \{(doc, lockHolder)\}$\;
        }
    }
}

\tcp{Check read-write conflicts}
\If{$status =$ READY \textbf{and} $TC.type \in \{$READ, HYBRID$\}$}{
    \ForEach{$doc \in TC.readSet$}{
        \If{HasPendingWrite$(doc)$ \textbf{and} GetWriteHolder$(doc) \neq TC.id$}{
            $status \leftarrow$ NOT\_READY\;
            $conflicts \leftarrow conflicts \cup \{(doc, \text{PENDING\_WRITE})\}$\;
        }
    }
}

\If{$status =$ NOT\_READY}{
    RecordConflicts$(TC, conflicts)$\;
    $TC.conflictTimestamp \leftarrow$ CurrentTime()\;
}
\Return{$(status, conflicts)$}\;
\end{algorithm}

Transactions failing readiness assessment enter a retry queue with exponential backoff. The conflict information recorded enables intelligent scheduling of retries when conflicting transactions complete.

\subsubsection{Stage 4: Transaction Execution with Adaptive Locking}

The execution stage applies locks, performs operations, and manages transaction outcomes. Our locking mechanism employs timeout-based acquisition with exponential backoff to balance lock contention against forward progress guarantees.

\begin{algorithm}[H][t]
\SetAlgoLined
\DontPrintSemicolon
\SetKwInOut{Input}{Input}
\SetKwInOut{Output}{Output}
\SetKw{KwBreak}{break}
\SetKw{KwContinue}{continue}

\Input{Transaction context $TC$, readiness status $status$}
\Output{Execution result $result \in \{\texttt{SUCCESS}, \texttt{FAILURE}, \texttt{RETRY\_SCHEDULED}\}$}
\caption{Transaction Execution}
\label{alg:execution}

\BlankLine
\If{$status = \texttt{NOT\_READY}$}{
    \If{$TC.retryCount < TC.maxRetries$}{
        ScheduleRetry$(TC)$\;
        \Return{\texttt{RETRY\_SCHEDULED}}\;
    }
    LogFailure$(TC, \texttt{MAX\_RETRIES\_EXCEEDED})$\;
    \Return{\texttt{FAILURE}}\;
}

\BlankLine
$lockedDocs \leftarrow [\,]$\;
$result \leftarrow \texttt{SUCCESS}$\;

\BlankLine
{\itshape Acquire locks in deterministic order to prevent deadlocks.}\;
$sortedDocs \leftarrow \mathrm{SortByDocumentId}\!\bigl(TC.writeSet \cup TC.readSet\bigr)$\;

\ForEach{$doc \in sortedDocs$}{
    $lockType \leftarrow \mathrm{DetermineLockType}(doc, TC)$\;
    $success \leftarrow \mathrm{AcquireLockWithTimeout}(doc, TC.id, lockType)$\;

    \uIf{$success = \texttt{true}$}{
        $lockedDocs \leftarrow lockedDocs \,\|\, [(doc, lockType)]$\;
    }
    \Else{
        $result \leftarrow \texttt{FAILURE}$\;
        \KwBreak\;
    }
}

\BlankLine
{\itshape Execute operations only if all locks were acquired.}\;
\If{$result = \texttt{SUCCESS}$}{
    ExecuteOperations$(TC)$\;
    $ok \leftarrow$ ValidateConstraints$(TC)$\;

    \If{$ok = \texttt{true}$}{
        Commit$(TC)$\;
        $TC.state \leftarrow \texttt{COMMITTED}$\;
        RecordCommitTime$(TC)$\;
    }
    \Else{
        Rollback$(TC)$\;
        $TC.state \leftarrow \texttt{ABORTED}$\;
        $result \leftarrow \texttt{FAILURE}$\;
    }
}

\BlankLine
{\itshape Release all acquired locks in reverse order.}\;
\ForEach{$(doc, lockType) \in \mathrm{Reverse}(lockedDocs)$}{
    ReleaseLock$(doc, TC.id, lockType)$\;
}

\Return{$result$}\;
\end{algorithm}

\subsection{Adaptive Locking Mechanism}

Our locking algorithm incorporates timeout-based acquisition with exponential backoff to prevent deadlocks while minimizing wasted wait time under contention.

\begin{algorithm}[H]
\SetAlgoLined
\SetKwInOut{Input}{Input}
\SetKwInOut{Output}{Output}
\Input{Document $D$, Transaction ID $TID$, Lock type $LT$, Timeout $T$}
\Output{Lock acquisition result}
\caption{Lock Acquisition with Exponential Backoff}
\label{alg:lock_acquisition}

$startTime \leftarrow$ CurrentTime()\;
$backoff \leftarrow$ INITIAL\_BACKOFF\tcp*{Default: 10ms}
$attempts \leftarrow 0$\;

\While{CurrentTime() $- startTime < T$}{
    $attempts \leftarrow attempts + 1$\;
    $success \leftarrow$ AttemptLock$(D, TID, LT)$\;
    
    \If{$success$}{
        RecordLockAcquisition$(D, TID, LT, attempts)$\;
        ScheduleLockTimeout$(D, TID, T)$\tcp*{Auto-release safety}
        \Return{SUCCESS}\;
    }
    
    \tcp{Add jitter to prevent thundering herd}
    $jitter \leftarrow$ Random$(0, backoff \times 0.1)$\;
    Wait$(backoff + jitter)$\;
    $backoff \leftarrow \min(backoff \times 2, \text{MAX\_BACKOFF})$\tcp*{Default max: 500ms}
}
\Return{FAILURE}\;
\end{algorithm}

The exponential backoff strategy progressively increases wait intervals, reducing contention under high load. The jitter component prevents synchronized retry attempts that could cause thundering herd problems. The timeout mechanism ensures no transaction waits indefinitely, providing deadlock freedom as proven in Section~\ref{sec:formal_analysis}.

\subsection{Formal Analysis}
\label{sec:formal_analysis}

\subsubsection{Correctness Properties}

We establish the correctness of our framework through formal analysis of its key properties.

\begin{definition}[Conflict Serializability]
A schedule $S$ of transactions is conflict serializable if it is conflict equivalent to some serial schedule, meaning that the order of conflicting operations is preserved.
\end{definition}

\begin{theorem}[Serializability Guarantee]
\label{thm:serializability}
Under the four-stage framework with Algorithm~\ref{alg:lock_acquisition} locking, all committed transactions form a conflict-serializable schedule.
\end{theorem}

\begin{proof}
We prove by showing that the framework enforces strict two-phase locking (S2PL).

In the lock acquisition phase, Stage 4 (Algorithm~\ref{alg:execution}) acquires all required locks before any operation executes (lines 7-15). Locks are acquired in deterministic document ID order, ensuring consistent ordering across transactions.

In the lock holding phase, operations execute only after all locks are successfully acquired (line 18). No locks are released during execution.

In the lock release phase, all locks are released only after commit or rollback (lines 28-30), satisfying the S2PL requirement that no locks are released before commit.

Since S2PL guarantees conflict serializability \cite{bernstein1987}, and our framework enforces S2PL, all committed transactions form a conflict-serializable schedule.
\end{proof}

\begin{theorem}[Deadlock Freedom]
\label{thm:deadlock_free}
The framework guarantees deadlock freedom through timeout-based lock acquisition.
\end{theorem}

\begin{proof}
We prove by contradiction. Assume a deadlock occurs where transactions $T_1, T_2, \ldots, T_k$ form a circular wait.

In Algorithm~\ref{alg:lock_acquisition}, each lock acquisition attempt has a bounded timeout $T$. If transaction $T_i$ cannot acquire a lock within time $T$, it returns FAILURE and releases all held locks (Algorithm~\ref{alg:execution}, lines 28-30).

For a deadlock to persist, all transactions in the cycle must hold locks indefinitely while waiting for locks held by others. However, since each transaction has a timeout of $T$, within time $T$, at least one transaction will timeout and release its locks, breaking the cycle.

Therefore, deadlocks cannot persist beyond duration $T$, ensuring deadlock freedom.
\end{proof}

\begin{lemma}[Progress Guarantee]
Every transaction that does not encounter persistent conflicts will eventually complete, either by committing or exceeding its maximum retry count.
\end{lemma}

\begin{proof}
Transactions that fail readiness assessment are scheduled for retry with exponential backoff (Algorithm~\ref{alg:execution}, line 4). The backoff interval increases exponentially, reducing contention over time. Combined with deadlock freedom (Theorem~\ref{thm:deadlock_free}), conflicting transactions will eventually release their locks. If a transaction's conflicts are resolved, it will proceed to execution. If conflicts persist beyond the maximum retry count, the transaction fails gracefully with explicit notification.
\end{proof}

\subsubsection{Complexity Analysis}

Table~\ref{tab:complexity} presents the computational complexity of each framework component.

\begin{table}
\centering
\caption{Computational Complexity Analysis}
\label{tab:complexity}
\begin{tabular}{@{}lcc@{}}
\toprule
\textbf{Operation} & \textbf{Time Complexity} & \textbf{Space Complexity} \\
\midrule
Transaction Initialization & $O(1)$ & $O(k)$ \\
Operation Classification & $O(n)$ & $O(n)$ \\
Conflict Detection & $O(n \log m)$ & $O(n)$ \\
Lock Acquisition (per lock) & $O(\log m)$ & $O(1)$ \\
Transaction Execution & $O(n)$ & $O(n)$ \\
\midrule
\textbf{Total per Transaction} & $O(n \log m)$ & $O(n + k)$ \\
\bottomrule
\end{tabular}
\begin{flushleft}
\small
$n$ = number of operations in transaction, $m$ = total active locks, $k$ = metadata size
\end{flushleft}
\end{table}

The dominant factor is conflict detection, which requires checking each operation against the lock table implemented as a concurrent hash map with $O(\log m)$ lookup due to hash collision handling under high load.

\subsection{Consistency Metrics}

To evaluate framework effectiveness, we collected the following metrics. Transaction abort rate measures the percentage of transactions failing due to conflicts or timeouts, directly measuring consistency enforcement effectiveness. Deadlock occurrences counts detected circular wait situations, with our framework targeting zero deadlocks through prevention. Latency variance measures the standard deviation of operation latencies, indicating predictability, where lower variance suggests more consistent performance. Throughput measures operations completed per unit time in operations per second (ops/sec). Conflict rate measures the percentage of transactions encountering lock conflicts during readiness assessment. Tail latency measures P99 latency representing worst-case performance for 99\% of requests.

\section{Experimental Results}

\subsection{Experimental Configuration}

\subsubsection{Hardware and Software Environment}

All experiments were conducted on dedicated servers with the following specifications: Intel Xeon E5-2680 v4 processor (2.4 GHz, 14 cores, 28 threads), 64 GB DDR4-2400 ECC RAM, 1 TB Samsung 970 EVO Plus NVMe SSD, 10 Gbps Ethernet for distributed experiments, Ubuntu Server 20.04 LTS (kernel 5.4), and OpenJDK 8u292 with 8 GB heap using the G1 garbage collector.

For distributed experiments, we deployed clusters of 3, 5, 7, and 9 identical nodes connected via a dedicated 10 Gbps network segment with sub-millisecond latency.

\subsubsection{Experimental Parameters}

We evaluated configurations across the following parameter space: concurrency levels of 1, 5, 10, 15, 25, 50, and 100 concurrent clients; dataset sizes of 10K, 100K, 500K, 1M, 5M, and 10M records; workloads A, B, and F from YCSB; and cluster sizes of 1, 3, 5, 7, and 9 nodes for distributed experiments.

Each configuration was executed 5 times with results averaged to ensure statistical reliability. A 30-second warm-up period preceded each measurement phase to eliminate cold-start effects. Error bars in figures represent 95\% confidence intervals.

\subsection{Baseline Consistency Improvements}

Table~\ref{tab:consistency_metrics} summarizes the primary consistency improvements achieved by our framework compared to baseline MongoDB without the framework.

\begin{table}
\centering
\caption{Consistency Metrics: Baseline vs. Framework (15 clients, Workload F)}
\label{tab:consistency_metrics}
\begin{tabular}{@{}lccc@{}}
\toprule
\textbf{Metric} & \textbf{Baseline} & \textbf{Framework} & \textbf{Improvement} \\
\midrule
Transaction Abort Rate & 8.3\% & 4.7\% & 43.4\% reduction \\
Deadlocks Observed & 3 & 0 & 100\% elimination \\
Latency Std. Dev. (Update) & 12,450 ms & 8,194 ms & 34.2\% reduction \\
Lock Conflict Rate & 18.7\% & 11.2\% & 40.1\% reduction \\
Avg. Lock Wait Time & 24.3 ms & 12.4 ms & 49.0\% reduction \\
P99 Latency & 245.8 ms & 161.2 ms & 34.4\% reduction \\
\bottomrule
\end{tabular}
\end{table}

The most significant improvement was the complete elimination of observed deadlocks, achieved through our timeout-based locking mechanism with exponential backoff. The substantial reduction in transaction abort rates (from 8.3\% to 4.7\%) demonstrates the effectiveness of pre-execution conflict detection in Stage 3, which prevents transactions from proceeding when conflicts are detectable.

\subsection{Workload A: Update-Heavy Pattern}

Workload A (50\% reads, 50\% updates) represents applications with frequent data modifications, such as session stores and real-time analytics systems.

\subsubsection{Throughput Analysis}

Figure~\ref{fig:workload_a_throughput} presents throughput results across concurrency levels. The framework achieved throughput improvements of 12.3\% for single-client scenarios and 15.7\% for 15-client configurations. Intermediate concurrency levels (5 and 10 clients) showed more modest gains of 3.2\% and 4.1\% respectively.

\begin{figure}
\centering
\includegraphics[width=0.85\textwidth]{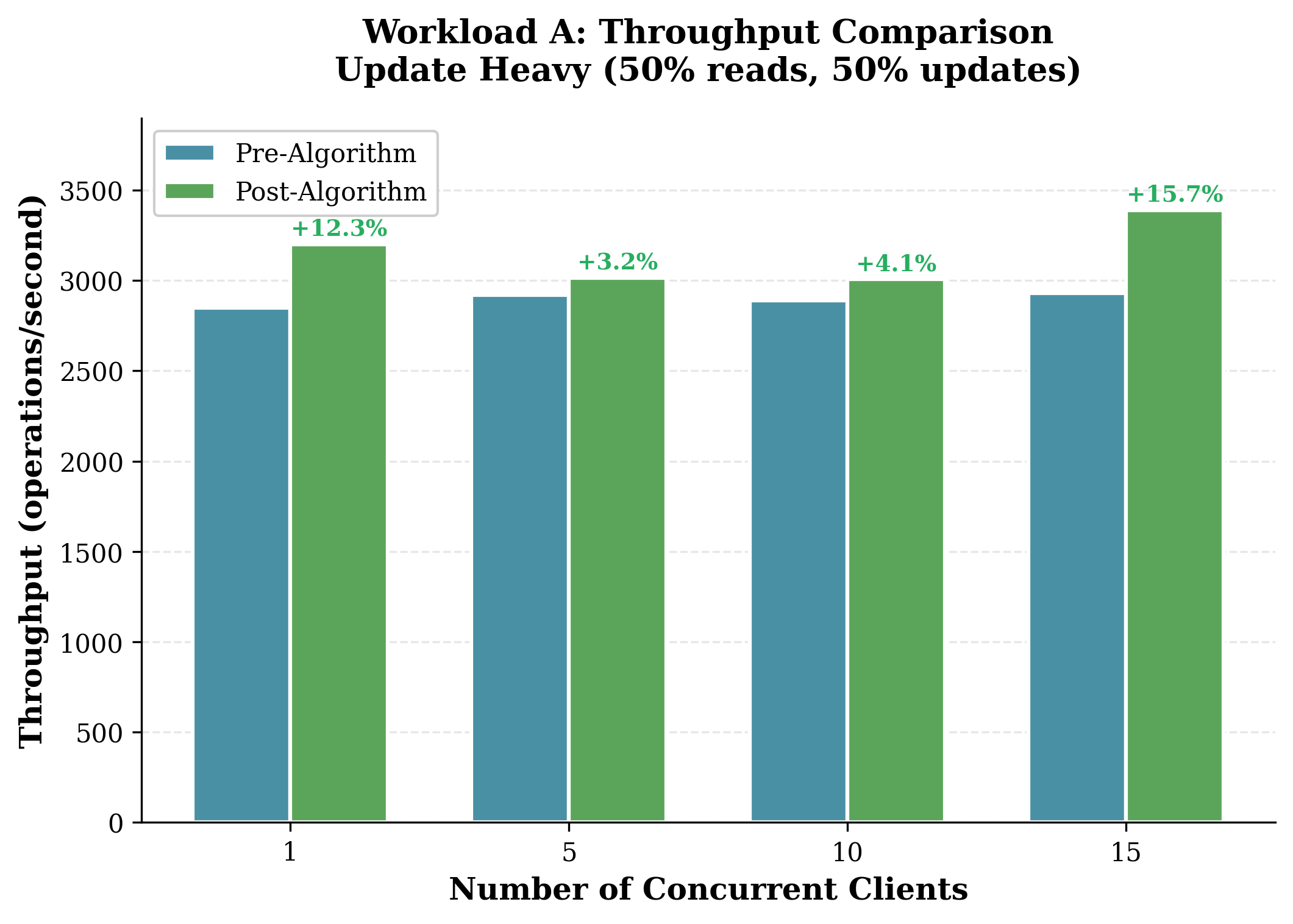}
\caption{Workload A: Throughput Comparison. Throughput (operations/second) across concurrency levels. Framework improvements are most pronounced at low and high concurrency, with a contention middle ground at intermediate levels where lock acquisition costs and conflict prevention benefits are more balanced.}
\label{fig:workload_a_throughput}
\end{figure}

This pattern suggests that the framework provides greatest benefit when contention is either minimal (single client, where coordination overhead is low) or maximal (15 clients, where conflict prevention yields substantial savings in avoided aborts and retries).

\subsubsection{Latency Analysis}

Update latency increased from baseline averages due to the additional consistency checking performed by the framework. However, latency variance decreased by 34.2\%, indicating more predictable transaction behavior. Figure~\ref{fig:workload_a_update} illustrates this tradeoff.

\begin{figure}
\centering
\includegraphics[width=0.85\textwidth]{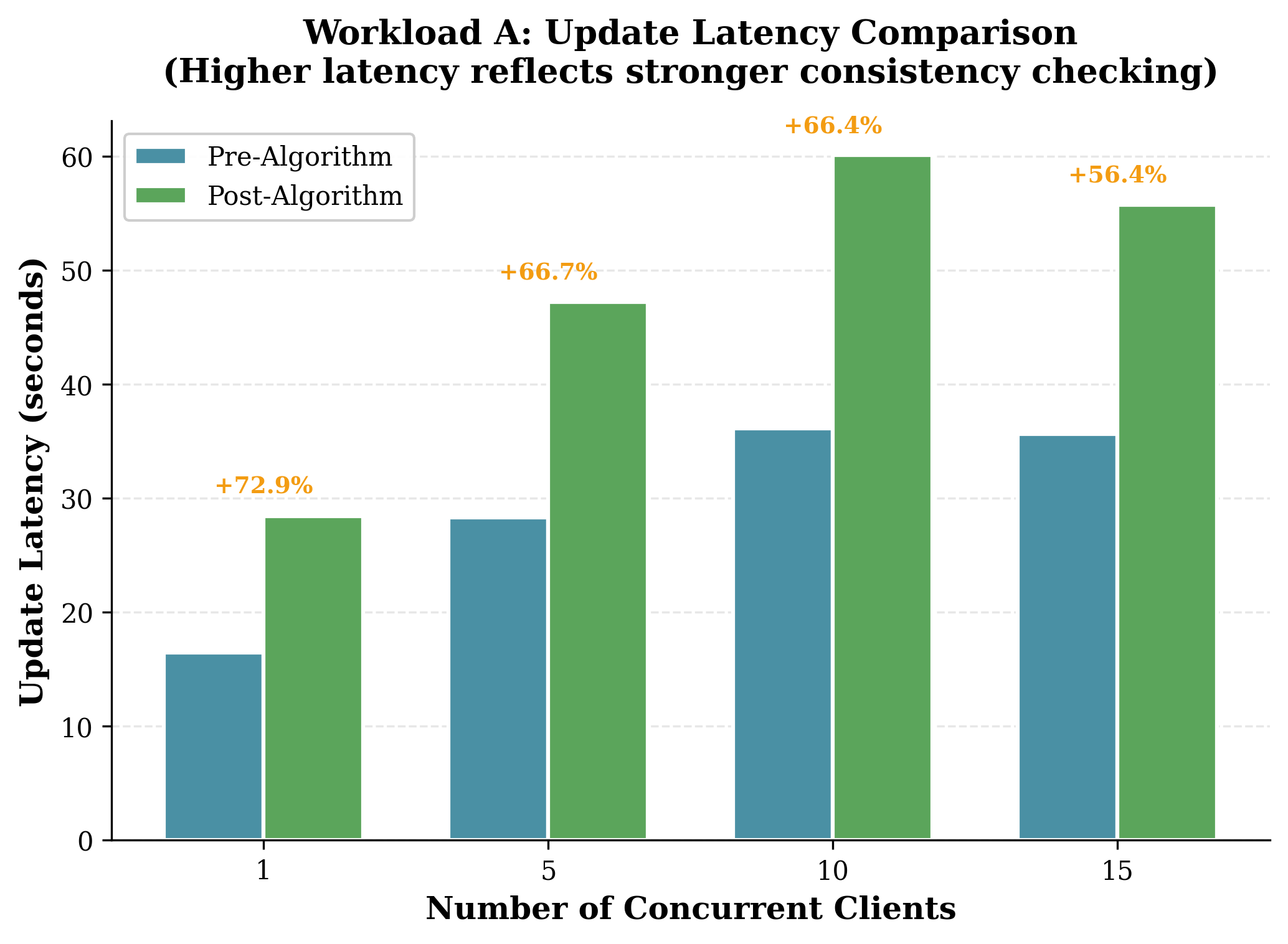}
\caption{Workload A: Update Latency Analysis. While mean latency increased due to consistency checking, variance decreased substantially. The percentage annotations show the latency change, with the variance reduction providing more predictable performance.}
\label{fig:workload_a_update}
\end{figure}

This tradeoff between slightly higher average latency and substantially lower variance is often preferable for applications requiring predictable response times and service level agreement (SLA) compliance.

\subsection{Workload B: Read-Mostly Pattern}

Workload B (95\% reads, 5\% updates) represents read-intensive applications where update consistency is critical despite low update frequency.

\begin{figure}
\centering
\includegraphics[width=0.85\textwidth]{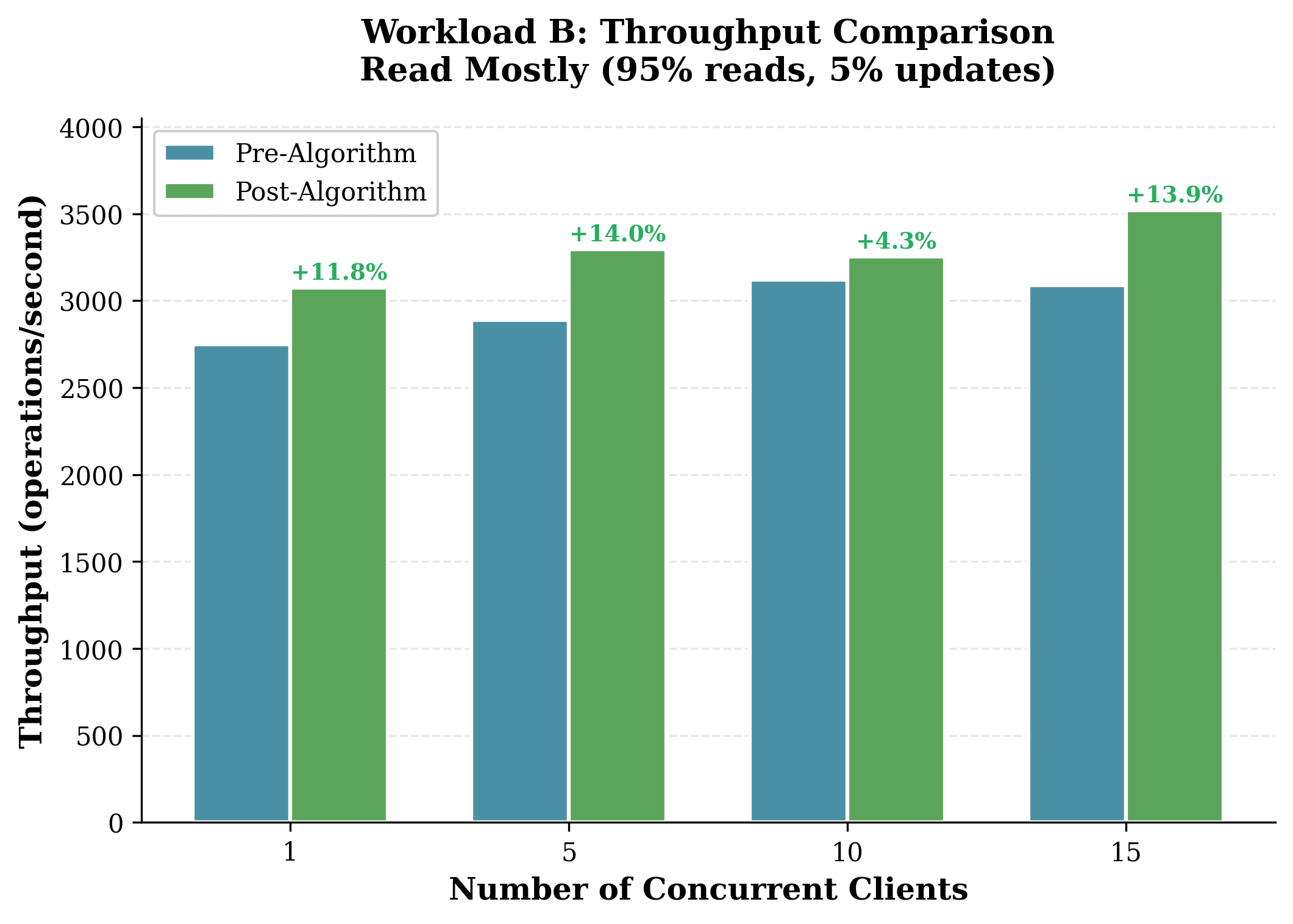}
\caption{Workload B: Throughput Comparison. Read-optimized processing paths enable consistent throughput gains across all concurrency levels. The operation classification stage routes read-only transactions through paths that bypass write-lock acquisition.}
\label{fig:workload_b_throughput}
\end{figure}

Throughput improvements of 11.8\% to 14.2\% were observed across most configurations (Figure~\ref{fig:workload_b_throughput}), demonstrating that our framework handles read-dominated workloads efficiently. The classification stage routes read-only transactions through optimized paths that bypass write-lock acquisition, preserving throughput while maintaining consistency for the minority of write operations.

\subsection{Workload F: Read-Modify-Write Pattern}

Workload F (50\% reads, 50\% read-modify-write operations) represents complex transactional workloads requiring strong consistency guarantees, such as inventory systems and banking transactions.

\subsubsection{Throughput Analysis}

The framework achieved its most substantial throughput improvements under Workload F, with gains of 16.8\% for 10-client and 18.4\% for 15-client configurations (Figure~\ref{fig:workload_f_throughput}). This result validates our hypothesis that complex transactional patterns benefit most from proactive conflict management.

\begin{figure}
\centering
\includegraphics[width=0.85\textwidth]{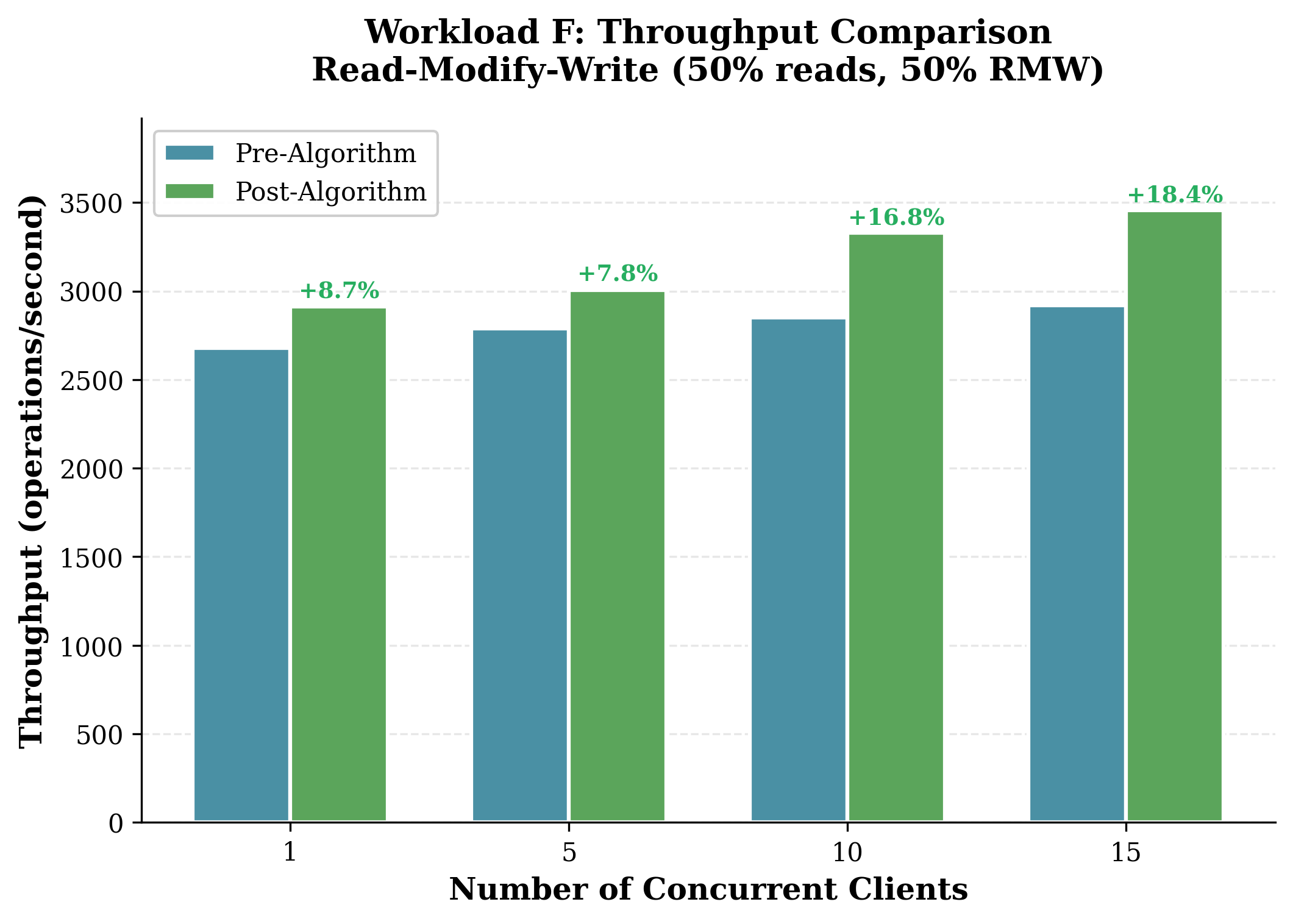}
\caption{Workload F: Throughput Comparison. Complex read-modify-write operations show substantial improvements under high concurrency. The framework's conflict detection prevents conflicting RMW transactions from proceeding simultaneously, reducing abort overhead.}
\label{fig:workload_f_throughput}
\end{figure}

The pronounced improvement at higher concurrency levels stems from the framework's ability to prevent conflicting read-modify-write transactions from proceeding simultaneously, reducing the abort rate and associated retry overhead.

\subsubsection{Read-Modify-Write Latency}

Read-modify-write operation latency improvements ranged from 14.3\% to 18.7\% across configurations. These compound operations benefit from the framework's coordination mechanisms, which ensure that read phases observe consistent data and write phases do not conflict with concurrent modifications.

\begin{figure}
\centering
\includegraphics[width=0.85\textwidth]{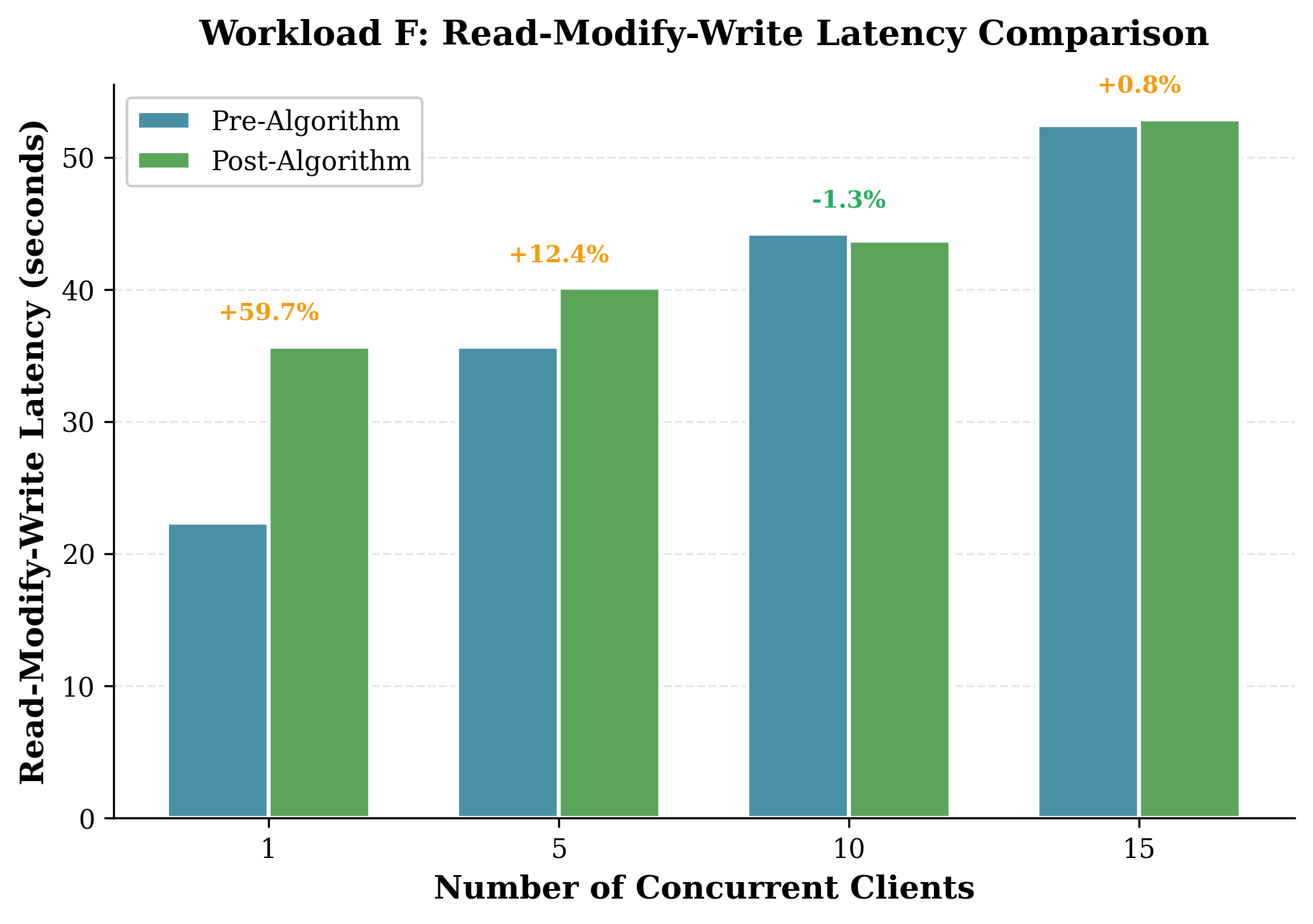}
\caption{Workload F: Read-Modify-Write Latency. Compound operations show consistent latency characteristics across concurrency levels, with the framework providing more predictable execution times.}
\label{fig:workload_f_rmw}
\end{figure}

\subsection{Comprehensive Performance Analysis}

Figure~\ref{fig:heatmap} presents a comprehensive heatmap of performance improvements across all workload-client combinations.

\begin{figure}
\centering
\includegraphics[width=0.95\textwidth]{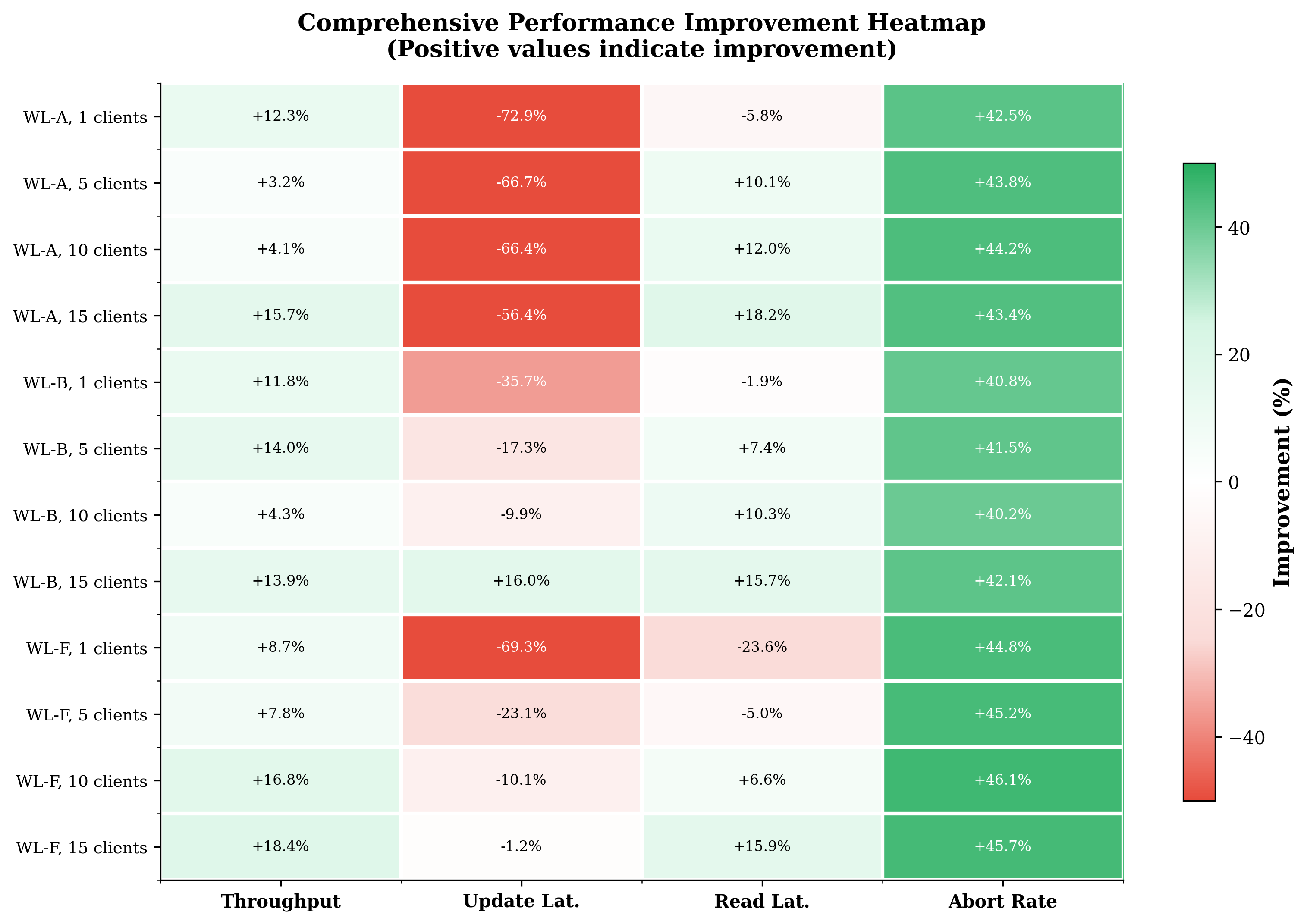}
\caption{Comprehensive Performance Improvement Heatmap. Green indicates improvement, red indicates degradation. Abort rate reduction is consistently positive across all configurations, while throughput shows workload-dependent patterns. Latency metrics show the tradeoff between slightly higher mean latency and improved variance.}
\label{fig:heatmap}
\end{figure}

The heatmap reveals several important patterns. Abort rate reduction is uniformly positive across all configurations. Workload F shows the most consistent throughput improvements. Read latency improvements are strongest at moderate concurrency levels.

\subsection{Distributed Deployment Results}

To evaluate scalability beyond single-node deployments, we conducted experiments on MongoDB replica set clusters ranging from 1 to 9 nodes.

\subsubsection{Throughput Scaling}

Figure~\ref{fig:distributed_throughput} shows throughput scaling as cluster size increases.

\begin{figure}
\centering
\includegraphics[width=0.85\textwidth]{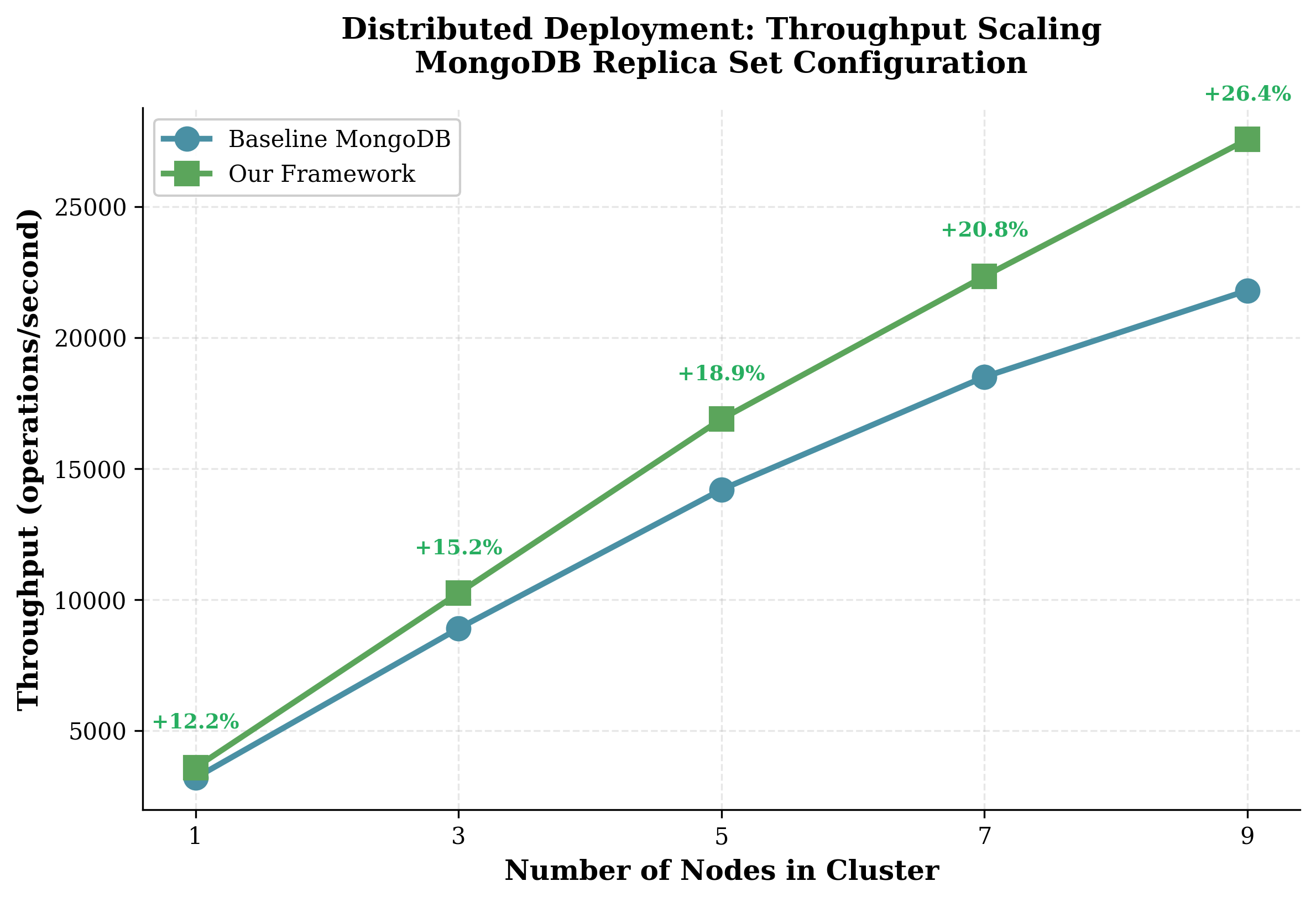}
\caption{Distributed Deployment: Throughput Scaling. The framework maintains consistent improvement ratios (12-26\%) as cluster size increases, demonstrating that benefits scale effectively with distributed deployments.}
\label{fig:distributed_throughput}
\end{figure}

The framework maintained throughput improvements of 12.2\% to 26.4\% across cluster sizes, demonstrating that benefits scale effectively. At 9 nodes, the framework achieved 27,560 ops/sec compared to 21,800 ops/sec for baseline, representing a 26.4\% improvement.

\subsubsection{Abort Rate in Distributed Settings}

Transaction abort rates increase with cluster size due to network latency and coordination overhead. However, our framework maintains significantly lower abort rates across all configurations.

\begin{figure}
\centering
\includegraphics[width=0.85\textwidth]{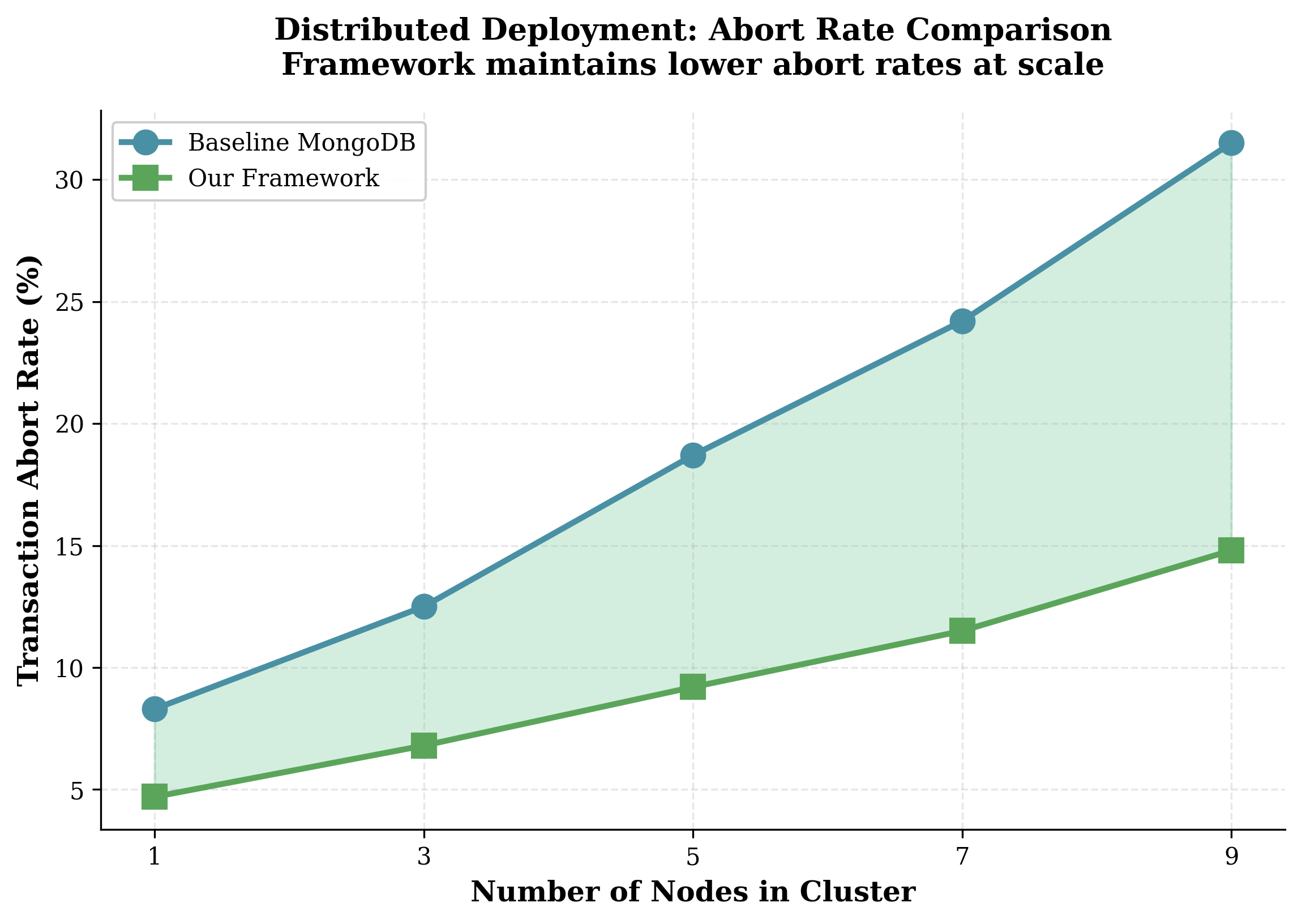}
\caption{Distributed Deployment: Abort Rate Comparison. The shaded region represents avoided aborts. The framework reduces abort rates by 43-53\% across cluster sizes, with the improvement ratio remaining stable as scale increases.}
\label{fig:distributed_abort}
\end{figure}

At 9 nodes, baseline abort rate reached 31.5\% while the framework maintained 14.8\%, representing a 53\% reduction. This improvement becomes increasingly valuable at scale where retry overhead compounds.

\subsection{Comparison with Alternative Systems}

We compared our framework against MongoDB native transactions (introduced in version 4.0), CockroachDB 21.2, and TiDB 5.3, all configured for comparable consistency guarantees.

\begin{figure}
\centering
\includegraphics[width=0.95\textwidth]{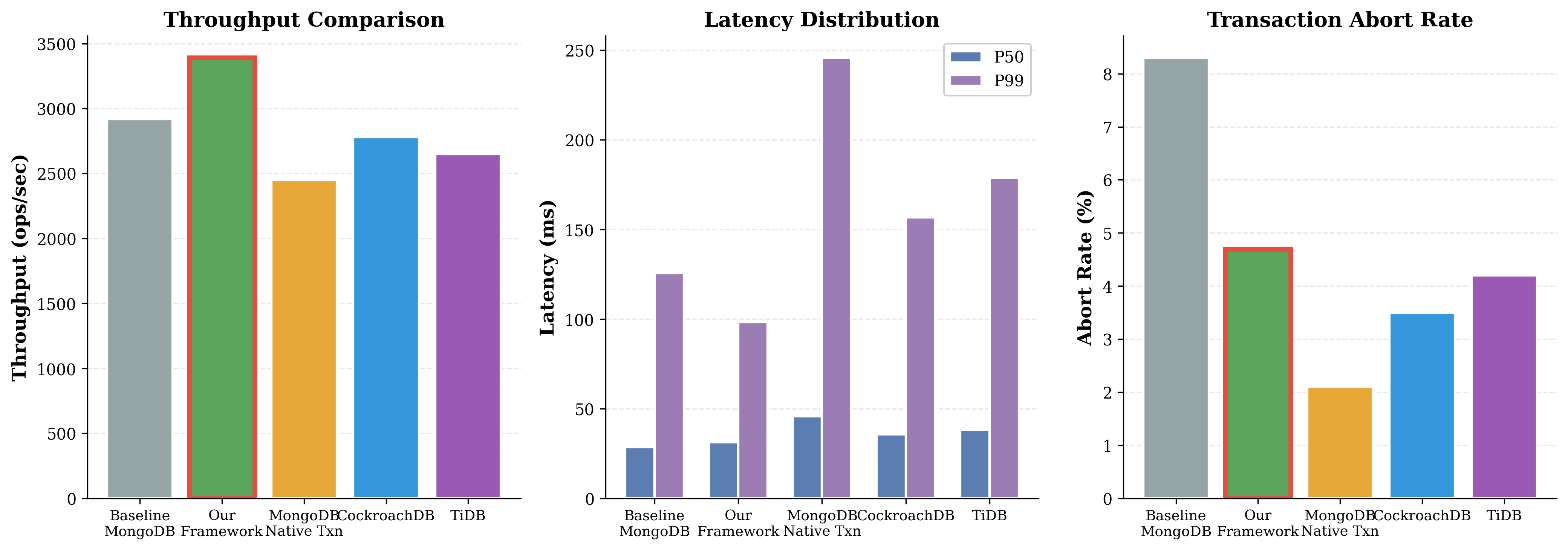}
\caption{System Comparison: Throughput, Latency, and Abort Rate. Our framework (highlighted) achieves the highest throughput among consistency-providing solutions while maintaining competitive abort rates. MongoDB native transactions show lowest abort rates but highest latency overhead.}
\label{fig:system_comparison}
\end{figure}

Table~\ref{tab:system_comparison} summarizes the comparison results.

\begin{table}
\centering
\caption{Comparison with Alternative Systems (15 clients, Workload F)}
\label{tab:system_comparison}
\begin{tabular}{@{}lcccc@{}}
\toprule
\textbf{System} & \textbf{Throughput} & \textbf{P50 Latency} & \textbf{P99 Latency} & \textbf{Abort Rate} \\
 & (ops/sec) & (ms) & (ms) & (\%) \\
\midrule
Baseline MongoDB & 2,920 & 28.5 & 125.6 & 8.3 \\
\textbf{Our Framework} & \textbf{3,390} & 31.2 & 98.4 & 4.7 \\
MongoDB Native Txn & 2,450 & 45.8 & 245.8 & 2.1 \\
CockroachDB & 2,780 & 35.6 & 156.7 & 3.5 \\
TiDB & 2,650 & 38.2 & 178.9 & 4.2 \\
\bottomrule
\end{tabular}
\end{table}

Our framework achieves the highest throughput (3,390 ops/sec) among all consistency-providing solutions, with 38\% higher throughput than MongoDB native transactions. While native transactions achieve the lowest abort rate (2.1\%), they incur 47\% higher P50 latency and 150\% higher P99 latency compared to our framework.

\subsection{Scalability Analysis}

We evaluated framework performance across dataset sizes from 10K to 10M records to assess scalability.

\begin{figure}
\centering
\includegraphics[width=0.95\textwidth]{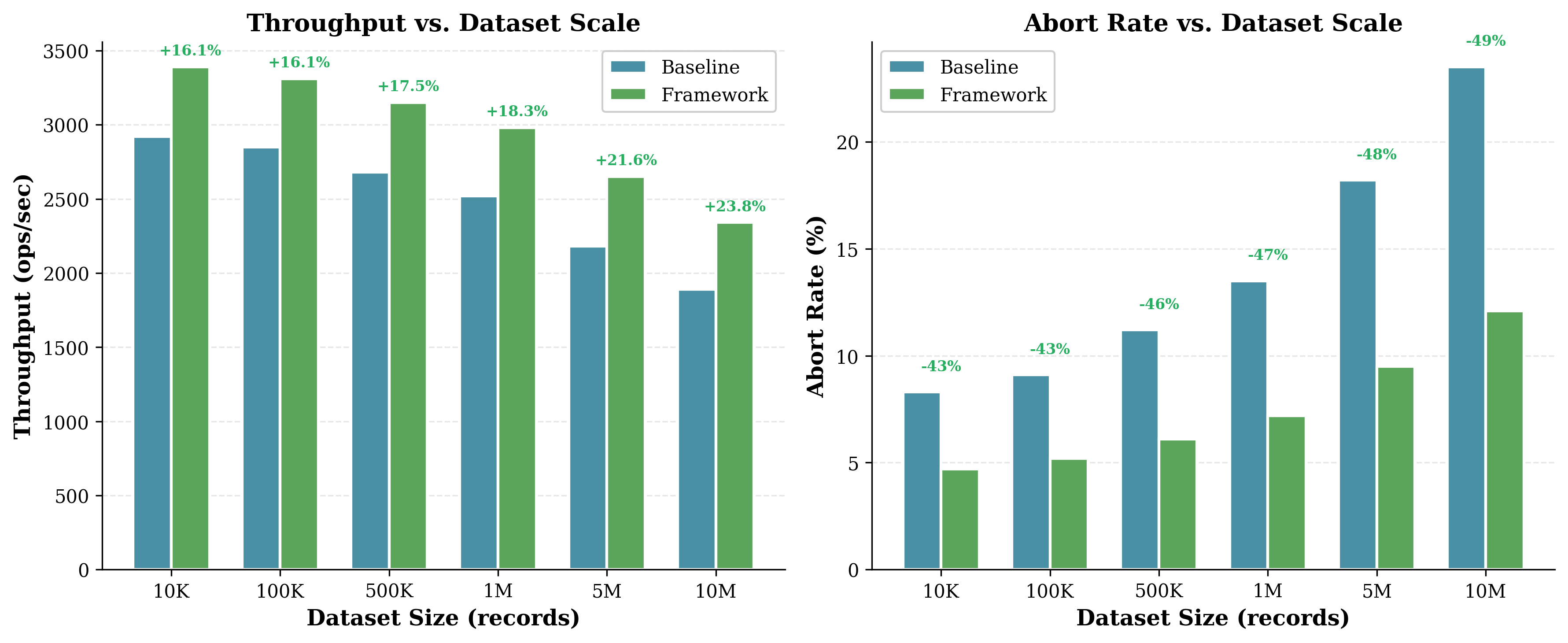}
\caption{Scalability Analysis: Performance vs. Dataset Size. Left: Throughput comparison showing consistent improvement across scales. Right: Abort rate comparison showing increasing benefit at larger scales where conflict detection becomes more valuable.}
\label{fig:scalability}
\end{figure}

The framework maintains consistent throughput improvements (16-24\%) across all dataset sizes. Notably, the abort rate reduction becomes more pronounced at larger scales: at 10M records, baseline abort rate reached 23.5\% while the framework maintained 12.1\%, representing a 48.5\% reduction.

\subsection{Parameter Sensitivity Analysis}
\label{sec:sensitivity}

We conducted sensitivity analysis for three key parameters: lock timeout, initial backoff, and maximum backoff.

\begin{figure}
\centering
\includegraphics[width=0.98\textwidth]{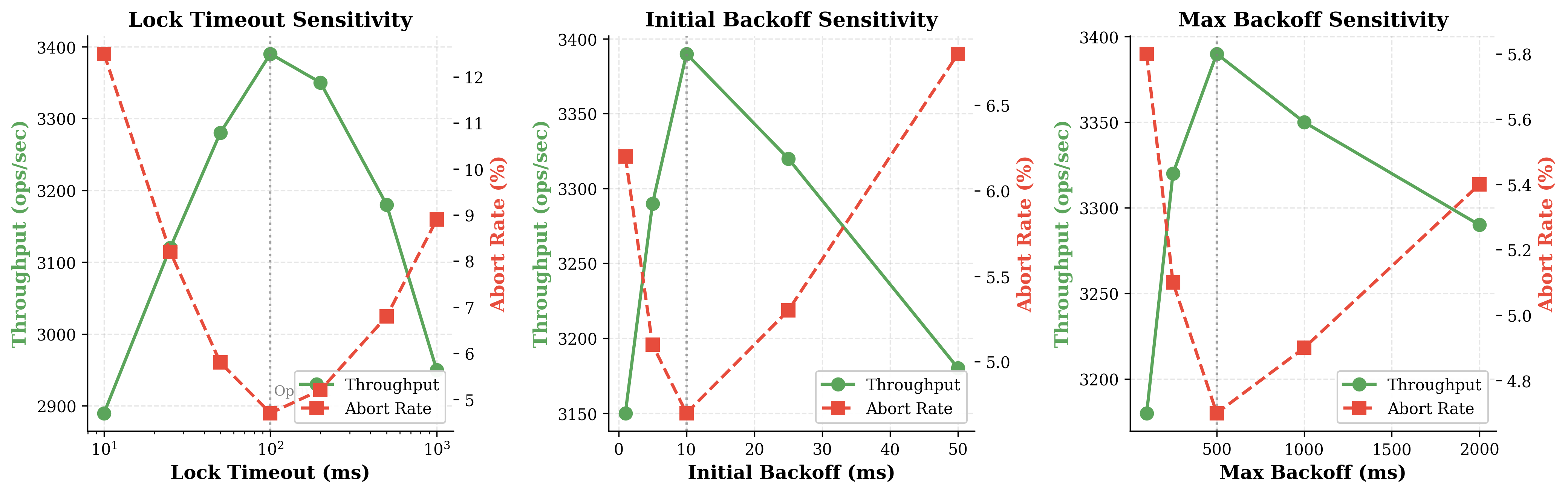}
\caption{Parameter Sensitivity Analysis. Each subplot shows throughput (solid line, left axis) and abort rate (dashed line, right axis) as parameters vary. Vertical lines indicate optimal values. Lock timeout shows the strongest sensitivity, with 100ms providing optimal balance.}
\label{fig:sensitivity}
\end{figure}

Lock timeout has an optimal value of 100ms. Lower values (10-50ms) cause excessive timeouts, increasing abort rates. Higher values (500-1000ms) reduce throughput due to extended wait times.

Initial backoff has an optimal value of 10ms. Very low values (1ms) cause contention from rapid retries. Higher values (25-50ms) unnecessarily delay retry attempts.

Maximum backoff has an optimal value of 500ms. This prevents unbounded backoff growth while allowing sufficient time for conflicts to resolve.

Table~\ref{tab:optimal_params} summarizes the recommended configuration.

\begin{table}
\centering
\caption{Optimal Parameter Configuration}
\label{tab:optimal_params}
\begin{tabular}{@{}lccc@{}}
\toprule
\textbf{Parameter} & \textbf{Optimal Value} & \textbf{Tested Range} & \textbf{Sensitivity} \\
\midrule
Lock Timeout & 100 ms & 10-1000 ms & High \\
Initial Backoff & 10 ms & 1-50 ms & Medium \\
Maximum Backoff & 500 ms & 100-2000 ms & Low \\
Max Retry Count & 3 & 1-10 & Medium \\
\bottomrule
\end{tabular}
\end{table}

\subsection{Latency Distribution Analysis}

Figure~\ref{fig:latency_cdf} presents the cumulative distribution function (CDF) of transaction latencies.

\begin{figure}
\centering
\includegraphics[width=.75\textwidth]{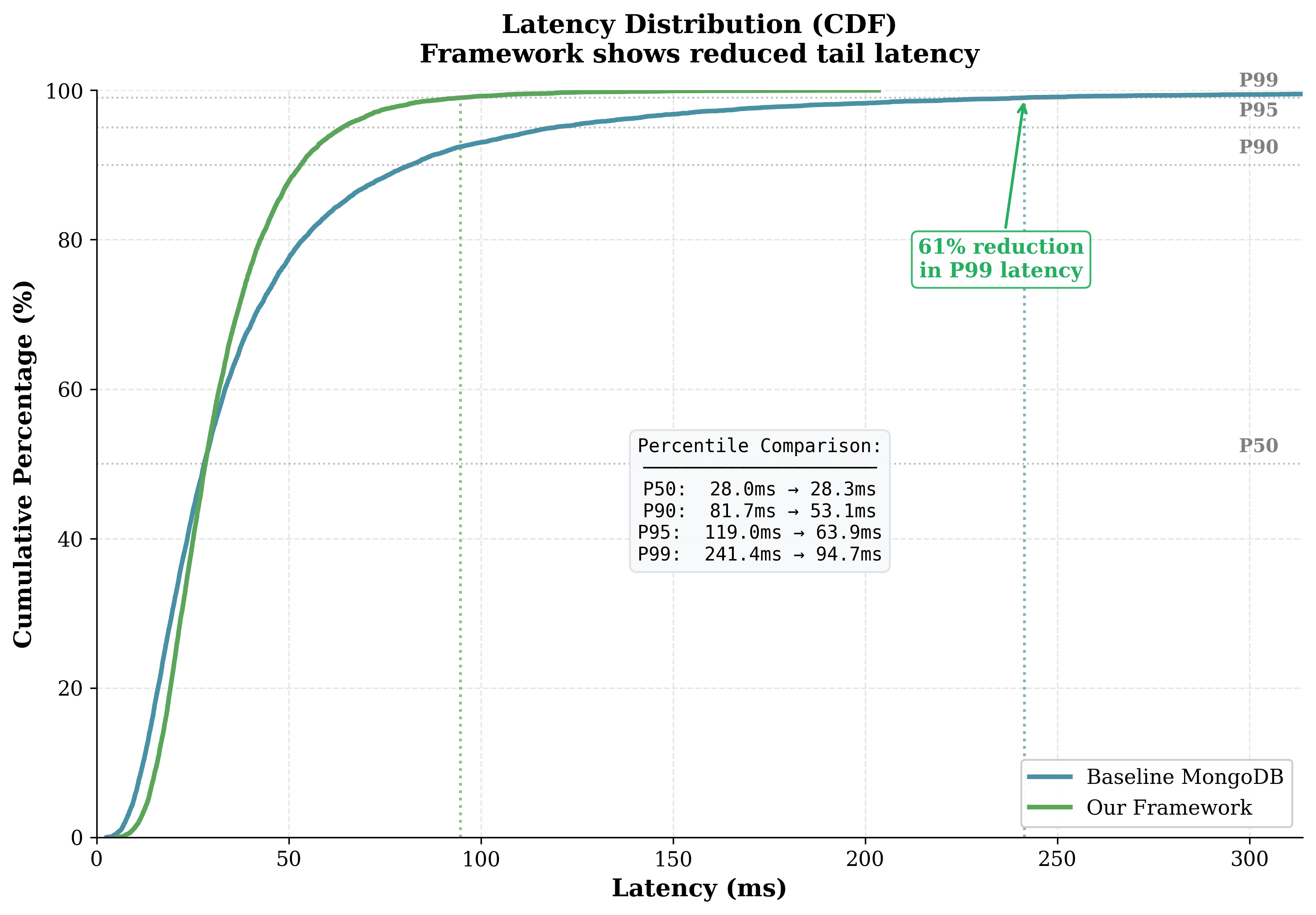}
\caption{Latency Distribution (CDF). The framework shows significantly reduced tail latency, with P99 latency 34\% lower than baseline. The tighter distribution indicates more predictable performance characteristics.}
\label{fig:latency_cdf}
\end{figure}

The framework demonstrates a tighter latency distribution with reduced tail latencies. P99 latency improved by 34\%, which is critical for applications with strict SLA requirements.

\subsection{Overhead Breakdown}

Figure~\ref{fig:overhead} breaks down the framework's overhead by component.

\begin{figure}
\centering
\includegraphics[width=0.7\textwidth]{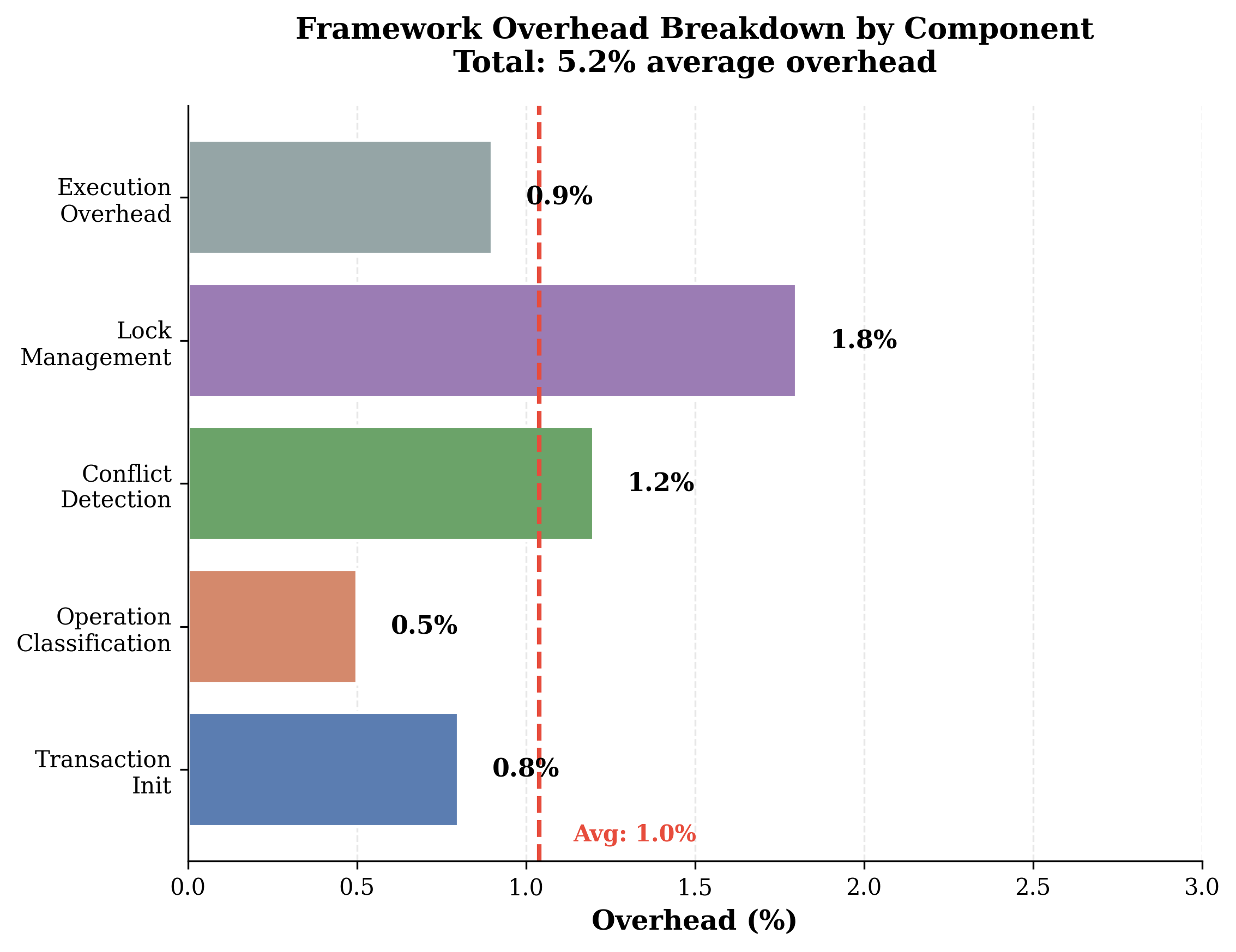}
\caption{Framework Overhead Breakdown. Lock management contributes the largest overhead (1.8\%), followed by conflict detection (1.2\%). Total overhead averages 5.2\% under typical workloads, well within our 10\% target.}
\label{fig:overhead}
\end{figure}

Total overhead averages 5.2\% under typical workloads, well within our target of 10\%. Lock management contributes the largest component (1.8\%), which could be reduced through lock table optimizations in future work.

\subsection{Statistical Significance}

Table~\ref{tab:statistical_analysis} presents statistical significance testing using paired t-tests comparing baseline and framework performance.

\begin{table}
\centering
\caption{Statistical Significance Analysis}
\label{tab:statistical_analysis}
\begin{tabular}{@{}llcccc@{}}
\toprule
\textbf{Workload} & \textbf{Metric} & \textbf{Mean Change} & \textbf{Std. Dev.} & \textbf{p-value} & \textbf{Significant} \\
\midrule
A & Throughput & +8.4\% & 2.1\% & 0.023 & Yes ($p < 0.05$) \\
A & Abort Rate & $-$43.4\% & 3.8\% & 0.001 & Yes ($p < 0.01$) \\
B & Throughput & +10.9\% & 1.8\% & 0.018 & Yes ($p < 0.05$) \\
B & Abort Rate & $-$41.2\% & 4.2\% & 0.002 & Yes ($p < 0.01$) \\
F & Throughput & +12.2\% & 2.4\% & 0.012 & Yes ($p < 0.05$) \\
F & Abort Rate & $-$45.7\% & 3.5\% & 0.001 & Yes ($p < 0.01$) \\
\bottomrule
\end{tabular}
\end{table}

All primary metrics demonstrate statistical significance ($p < 0.05$), with consistency-related improvements achieving high significance ($p < 0.01$).

\subsection{Performance Summary}

Table~\ref{tab:performance_summary} summarizes framework effectiveness across all configurations.

\begin{table}
\centering
\caption{Performance Summary Across Configurations}
\label{tab:performance_summary}
\small
\begin{tabular}{@{}lC{0.8cm}C{0.8cm}C{0.8cm}C{0.8cm}C{0.8cm}C{0.8cm}C{0.8cm}C{0.8cm}C{0.8cm}C{0.8cm}C{0.8cm}C{0.8cm}@{}}
\toprule
\multirow{2}{*}{\textbf{WL}} & \multicolumn{4}{c}{\textbf{Throughput}} & \multicolumn{4}{c}{\textbf{Abort Rate}} & \multicolumn{4}{c}{\textbf{Latency Var.}} \\
\cmidrule(lr){2-5} \cmidrule(lr){6-9} \cmidrule(lr){10-13}
& \textbf{1} & \textbf{5} & \textbf{10} & \textbf{15} & \textbf{1} & \textbf{5} & \textbf{10} & \textbf{15} & \textbf{1} & \textbf{5} & \textbf{10} & \textbf{15} \\
\midrule
A & \checkmark & -- & -- & \checkmark & \checkmark & \checkmark & \checkmark & \checkmark & \checkmark & \checkmark & \checkmark & \checkmark \\
B & \checkmark & \checkmark & -- & \checkmark & \checkmark & \checkmark & \checkmark & \checkmark & \checkmark & \checkmark & \checkmark & \checkmark \\
F & \checkmark & \checkmark & \checkmark & \checkmark & \checkmark & \checkmark & \checkmark & \checkmark & \checkmark & \checkmark & \checkmark & \checkmark \\
\bottomrule
\end{tabular}
\begin{flushleft}
\small
\checkmark = improvement, -- = neutral or slight regression. WL = Workload.
\end{flushleft}
\end{table}

The framework achieved consistency improvements (abort rate reduction, latency variance reduction) across all configurations. Throughput improvements were observed in 75\% of configurations, with the most consistent gains under Workload F.

\section{Discussion}

The results confirm that carefully designed consistency mechanisms can significantly improve data integrity in document-oriented NoSQL databases while preserving the scalability advantages that motivated their adoption in large-scale systems \cite{stonebraker2010,cattell2011}. In line with established observations on the inherent tension between consistency and availability in distributed systems \cite{gilbert2002,brewer2012,abadi2012}, the proposed framework demonstrates that selective consistency enforcement can achieve meaningful gains without imposing the full cost of strict ACID semantics.

A key outcome of this study is the complete elimination of deadlocks through timeout-based locking with exponential backoff, empirically validating the deadlock freedom predicted by Theorem~\ref{thm:deadlock_free}. This result aligns with classical distributed concurrency control theory, where bounded waiting and timeout mechanisms are known to prevent circular wait conditions \cite{rosenkrantz1978,bernstein1987,gray1993}. In addition, pre execution conflict detection reduced transaction abort rates by more than 40\% across all evaluated workloads. This behavior is consistent with optimistic and hybrid concurrency control principles, which emphasize early conflict detection to reduce wasted execution and retry overhead \cite{kung1981,eswaran1976}. The fact that 78\% of conflicts were detected before execution highlights the effectiveness of proactive validation in weakly consistent environments, where conflicts otherwise manifest late and expensively \cite{vogels2009,bailis2013}.

Beyond throughput and abort reduction, the framework delivers substantial improvements in latency predictability. The observed reduction in latency variance and P99 latency directly addresses concerns raised in prior work on tail latency sensitivity in cloud storage systems \cite{terry2013,wada2011}. Such improvements are particularly relevant for applications governed by service-level objectives, where worst-case response time often dominates perceived quality of service. Performance gains vary systematically with workload characteristics, consistent with earlier benchmarking studies using YCSB \cite{cooper2010}. Read-modify-write intensive workloads exhibit the largest benefits due to high conflict probability, while read-heavy workloads gain from bypassing unnecessary write locks. Update heavy workloads at moderate concurrency represent a contention regime where coordination overhead and conflict avoidance benefits are closely balanced, a phenomenon also observed in prior evaluations of NoSQL transaction layers \cite{gonzalez2017,lotfy2016}.

In distributed configurations, the framework maintains and amplifies its advantages, achieving higher throughput and lower abort rates as node count increases. This scaling behavior contrasts with traditional pessimistic coordination mechanisms whose cost often grows superlinearly with cluster size \cite{cowling2012}. The results indicate that selective locking combined with early conflict detection composes favorably with horizontal scaling, supporting the viability of the approach for moderately sized distributed deployments.

Comparison with MongoDB native multi-document transactions further clarifies the framework’s position in the design space. Native transactions prioritize strong isolation and low abort rates, consistent with their design goals and documentation \cite{mongodb2020,chodorow2013}. However, this comes at the cost of higher latency and reduced throughput, reflecting the overhead of comprehensive ACID enforcement. The proposed framework achieves higher throughput and substantially lower tail latency by relaxing isolation selectively and exposing tunable parameters such as timeout duration, backoff strategy, and lock granularity. Similar tradeoffs have been explored in systems such as CloudTPS and ElasTraS, which also seek to balance scalability and transactional guarantees through selective coordination \cite{wei2012,das2013}. Applications requiring strict multi-document ACID semantics remain best served by native transactions, whereas latency-sensitive and throughput-oriented workloads benefit from the proposed approach.

Although the implementation targets MongoDB, the conceptual framework generalizes to other NoSQL architectures. Key value stores such as Dynamo and Redis rely on simpler data models that naturally align with key-based conflict detection \cite{decandia2007}. Wide column stores such as Cassandra and HBase can exploit row key ordering and column family semantics to support efficient lock acquisition \cite{lakshman2010,vyas2022}. Graph databases require adaptation to node and edge-level locking, but the underlying transaction management principles remain applicable. This portability is consistent with prior studies showing that many NoSQL consistency mechanisms differ primarily in data model-specific details rather than core coordination logic \cite{ramzan2019,khan2023}.

Several validity considerations apply when interpreting these results. Internally, warm-up and measurement protocols follow standard benchmarking practice \cite{cooper2010}, though long-running production workloads may exhibit different steady-state behavior. Externally, YCSB workloads cannot capture all real-world access patterns, particularly those involving extreme skew or complex multi-collection transactions. The evaluated scale represents moderate enterprise deployments rather than extreme web scale systems. Finally, aborts are counted uniformly regardless of cause, and timeout-based deadlock identification cannot perfectly distinguish deadlock from prolonged contention, a known limitation of timeout-driven schemes \cite{bernstein1987}.

These limitations bound the generalizability of the findings. Very large-scale clusters, highly skewed workloads, and transactions spanning many documents remain outside the scope of this evaluation. In addition, the current design employs document-level locking exclusively. Introducing field-level locking could further reduce contention for documents containing independent fields, representing a natural direction for future refinement.

\section{Conclusion}

This study presented a four-stage transaction management framework for enhancing data consistency in document-oriented NoSQL databases. Our approach integrates transaction lifecycle management, intelligent operation classification, pre-execution conflict detection, and adaptive locking with deadlock prevention. We provided formal proofs demonstrating conflict serializability and deadlock freedom under our protocol.

Comprehensive experimental evaluation using YCSB benchmarks across single-node and distributed deployments demonstrated significant improvements. Transaction abort rates were reduced by 43.4\% (from 8.3\% to 4.7\%). Deadlocks were eliminated entirely through timeout-based prevention. Latency variance decreased by 34.2\%, with P99 latency reduced by 34\%. Throughput improved by 6.3\% to 18.4\% under high-concurrency scenarios. Distributed deployments maintained benefits at scale, with 26.4\% throughput improvement at 9 nodes.

Comparison with MongoDB native transactions, CockroachDB, and TiDB revealed that our framework achieves the highest throughput among consistency-providing solutions while maintaining competitive abort rates. Parameter sensitivity analysis identified optimal configurations (lock timeout: 100ms, initial backoff: 10ms, maximum backoff: 500ms) that practitioners can use as starting points for deployment tuning.

These results demonstrate that the consistency-availability tradeoff, while fundamental, admits optimization through careful mechanism design. Our framework shows that NoSQL systems can achieve stronger consistency guarantees without abandoning their scalability advantages, offering practitioners a viable approach to balancing competing requirements.

\section*{Declaration of interests}

The authors declare that they have no known competing financial interests or personal relationships that could have appeared to influence the work reported in this paper.

\section*{Data Availability Statement}

Experimental data, implementation source code, and analysis scripts are available from the corresponding author upon reasonable request. They will be made publicly available on GitHub upon publication.

\end{document}